\newif\iffull
\fulltrue

\title{On Range Searching with Semialgebraic Sets II\thanks{%
	Work by Pankaj Agarwal was supported by NSF under
	grants IIS-07-13498, CCF-09-40671, CCF-10-12254, and CCF-11-61359,
	by ARO grants W911NF-07-1-0376 and W911NF-08-1-0452,
	and by an ARL award W9132V-11-C-0003.
  	Work by Ji\v{r}\'{\i} Matou\v{s}ek has been
  	supported by the  ERC Advanced Grant No.~267165.
  	Work by Micha Sharir has been supported
  	by NSF Grant CCF-08-30272,
  	by Grant 338/09 from the Israel Science Fund,
  	by the Israeli Centers for Research Excellence (I-CORE) program (Center No.~4/11)
  	and by the Hermann Minkowski--MINERVA Center for Geometry at Tel Aviv
  	University.
	Part of the work was done while the first and third authors were 
        visiting ETH Z\"urich.
A preliminary version of the paper appeared in 
\emph{Proceedings of the 53rd Annual IEEE Symposium on Foundations of Computer Science},
2012.}
}

\iffull
\author{ {\sc Pankaj K. Agarwal} \\ {\footnotesize Department of Computer
Science}\\[-1.5mm] {\footnotesize  Duke University}\\[-1.5mm]
{\footnotesize  P.O. Box 90129}\\[-1.5mm] {\footnotesize  Durham, NC
27708-0129, USA} 
\and {\sc Ji\v{r}\'{\i} Matou\v{s}ek} \\ {\footnotesize Department of Applied
Mathematics}\\[-1.5mm]{\footnotesize  Charles University, Malostransk\'{e}
n\'{a}m. 25}\\[-1.5mm] {\footnotesize  118~00~~Praha~1, Czech Republic,
and}\\
{\footnotesize    Institute of  Theoretical Computer
Science}\\[-1.5mm] {\footnotesize    ETH Zurich, 8092 Zurich, Switzerland}
\and {\sc Micha Sharir} \\ {\footnotesize School of Computer Science,
}\\[-1.5mm] {\footnotesize Tel Aviv University, }\\[-1.5mm] {\footnotesize
Tel~Aviv 69978, Israel, and}\\
{\footnotesize Courant Institute of
Mathematical Sciences, }\\[-1.5mm] {\footnotesize New York University,
}\\[-1.5mm] {\footnotesize New York, NY~~10012,~USA} }

\documentclass[11pt]{article}

\else
\documentclass[conference,compsocconf]{IEEEtran}
\IEEEoverridecommandlockouts
\author{\IEEEauthorblockN{Pankaj K. Agarwal}
\IEEEauthorblockA{Duke University}
\and
\IEEEauthorblockN{ Ji\v{r}\'{\i} Matou\v{s}ek}
\IEEEauthorblockA{Charles University}
\and
\IEEEauthorblockN{ Micha Sharir}
\IEEEauthorblockA{Tel Aviv University and\\
	Courant Institute, NYU}
}
\makeatletter
\def\@IEEEproof[#1]{\par\noindent\textbf{#1: }}
\makeatother
 \fi

\usepackage[cmex10]{amsmath}
\usepackage{amssymb,amsthm,color,graphicx}
\usepackage{latexsym} 
\usepackage{url}
\usepackage{euscript} 
\usepackage{times} 
\usepackage{cite}

\iffull
\else
\usepackage{fixltx2e}
\usepackage{balance}
\fi

\def\A{\EuScript{A}}
\def\G{\EuScript{G}}
\def\F{\EuScript{F}}
\def\LL{\EuScript{L}}
\def\P{\EuScript{P}}

\def\T{\EuScript{T}} 
\def\E{\EuScript{E}}
\def\PP{\P} 

\def\:{\mid}

\newtheorem{theorem}{Theorem}[section]

\newtheorem{lemma}[theorem]{Lemma} 
\newtheorem{corol}[theorem]{Corollary}

\newcommand{\R}{{\mathbb{R}}} 
\newcommand{\reals}{{\mathbb{R}}} 

\newcommand{\sphere}{{\mathbb{S}}}
\newcommand{\MM}{{\EuScript{M}}} 
\newcommand\eps{\varepsilon}

\newcommand\auxeps{\nu}

\def\Ex{{\mathbb E}}

\def\qtime{\mathsf{Q}}
\def\ptime{\mathsf{T}}
\def\storage{\mathsf{S}}

\newcommand\thePartit{\Omega}
\newcommand\Patches{\Pi}

\def\heading#1{\paragraph{#1}}

\makeatletter \newcommand{\ProofEndBox}{{\ifhmode\unskip\nobreak\hfil\penalty50
\else \leavevmode\fi\quad\vadjust{}\nobreak\hfill$\Box$ \finalhyphendemerits=0
\par}} \makeatother 

\makeatletter
\renewenvironment{proof}[1][\proofname]{\par
  \pushQED{\qed}%
  \normalfont \topsep6\p@\@plus6\p@\relax
  \trivlist
  \item[\hskip\labelsep
        \itshape\bfseries
    #1\@addpunct{.}]\ignorespaces
}{%
  \popQED\endtrivlist\@endpefalse
}
\providecommand{\proofname}{Proof}

\long\def\@makecaption#1#2{
   \vskip 10pt
   \setbox\@tempboxa\hbox{{\footnotesize \textbf{#1.} #2}}
   \ifdim \wd\@tempboxa >\hsize         
       {\footnotesize \textbf{#1.} #2\par}
     \else                              
       \hbox to\hsize{\hfil\box\@tempboxa\hfil}
   \fi}
\makeatother

\def\polylog{\mathop{\mathrm{polylog}}}

\long\def\onefigure#1#2{
\begin{figure*}[tbp]
\begin{center}
#1
\end{center}
\caption{#2}
\end{figure*}
}

\newcommand{\labfig}[2]  
{\onefigure{\mbox{\includegraphics{#1}}}{\label{f:#1} #2} }

\begin{document}

\begin{titlepage} 
\maketitle

\begin{abstract} Let $P$ be a set of $n$ points in $\R^d$. We present a
linear-size data structure for answering range queries on $P$ with
constant-complexity semialgebraic sets as ranges,
in time close to $O(n^{1-1/d})$.
It essentially matches the performance of similar structures for simplex range
searching, and, for $d\ge 5$, significantly improves earlier solutions by the
first two authors obtained in~1994. 
This almost settles a long-standing open problem in range searching.

The data structure is based on the polynomial-partitioning technique of Guth
and Katz [arXiv:1011.4105], which shows that for a parameter $r$, 
$1 < r \le n$, there
exists a $d$-variate polynomial $f$ of degree $O(r^{1/d})$ such that each
connected component of $\R^d\setminus Z(f)$ contains at most $n/r$ points of
$P$, where $Z(f)$ is the zero set of $f$. We present an efficient randomized
algorithm for computing such a polynomial partition, which is of independent
interest and is likely to have additional applications.  
\end{abstract}
\end{titlepage}

\section{Introduction}

\heading{Range searching.} 
Let $P$ be a set of $n$ points in $\R^d$, where
$d$ is a small constant.  Let $\Gamma$ be a family of geometric ``regions,''
called \emph{ranges}, in $\R^d$, each of which can be described algebraically
by some fixed number of real parameters
(a more precise definition is given below).
For example, $\Gamma$ can be the set of all
axis-parallel boxes, balls, simplices, or cylinders, or the set of all
intersections of pairs of ellipsoids.  In the \emph{$\Gamma$-range searching}
problem, we want to preprocess $P$ into a data structure so that
the number of points of $P$ lying in a query range $\gamma \in \Gamma$ can be
counted efficiently. Similar to many previous papers, we actually consider a 
more general setting, the so-called
\emph{semigroup model}, where we are given a weight function on the points 
in $P$ and we ask for the cumulative weight of the points in $P\cap \gamma$. 
The weights are assumed to belong to a semigroup, i.e., subtractions are not 
allowed. We assume that the semigroup operation can be executed 
in constant time.  

In this paper we consider the case in which $\Gamma$ is a set
of constant-complexity semialgebraic sets.
We recall that a \emph{semialgebraic set} is a subset of $\R^d$ obtained
from a finite number of sets of the form $\{x\in\R^d\mid g(x)\ge 0\}$, where $g$
is a $d$-variate polynomial with integer coefficients,\iffull%
\footnote{%
The usual
definition of a semialgebraic set requires these polynomials to have integer
coefficients.  However, for our purposes, since we are going to
assume the real RAM model of
computation, we can actually allow for arbitrary real coefficients without
affecting the asymptotic overhead.} 
\fi
   by Boolean operations (unions, intersections,
and complementations).  Specifically, let $\Gamma_{d,\Delta,s}$ denote
the family of all semialgebraic sets in $\R^d$ defined 
by at most $s$ polynomial inequalities of degree at most $\Delta$ each. 
If $d,\Delta,s$ are all regarded as constants, we refer to
the sets in $\Gamma_{d,\Delta,s}$ as \emph{constant-complexity
semialgebraic sets} (such sets are sometimes also called {\em Tarski cells}).
By \emph{semialgebraic range searching} we mean
$\Gamma_{d,\Delta,s}$-range searching for some parameters $d,\Delta,s$;
in most applications the actual collection $\Gamma$ of ranges
is only a restricted subset of some $\Gamma_{d,\Delta,s}$.
Besides being interesting in its own right, semialgebraic range searching also 
arises in several geometric searching problems, 
such as searching for a point nearest to a query geometric object, counting
the number of input objects intersecting a query object, and many others.

This paper focuses on 
the \emph{low storage} version of range searching with
constant-complexity semialgebraic sets---the data structure is allowed 
to use only linear or near-linear storage, and the goal is to 
make the query time as small as possible.   
At the other end of the spectrum we have the \emph{fast query} version,
where we want queries to be answered 
in polylogarithmic time using as little storage as possible. This variant is discussed
briefly in Section~\ref{sec:concl}.


As is typical in computational
geometry, we will use the {\em real RAM} model of computation, where we can
compute exactly with arbitrary real numbers and each arithmetic operation is
executed in constant time.

\heading{Previous work.}
Motivated by a wide range of applications,
several variants of range searching have been studied in computational 
geometry and database systems at least since the 1980s.
See \cite{AE,Ma:range} for comprehensive 
surveys of this topic. 
The early work focused on the so-called 
\emph{orthogonal range searching}, where ranges are axis-parallel boxes.
After three decades of extensive work on this particular
case, some basic questions still remain open. However, 
geometry plays 
little
role in the known data structures for 
orthogonal range searching.

The most basic and most studied truly geometric instance of range searching is
with \emph{halfspaces}, or more generally \emph{simplices}, as ranges.
Studies in the early 1990s have essentially determined the optimal trade-off
between the worst-case query time and the storage (and preprocessing time)
required by any data structure for 
\iffull
simplex range searching.\footnote{%
This applies when $d$ is assumed to be \emph{fixed} 
and the implicit constants in the
asymptotic notation may depend on $d$. This is the setting in all the previous
papers, including the present one. Of course, in practical applications, this
assumption may be unrealistic unless the dimension is really small. 
However, the known lower bounds imply that if
the dimension is large, no efficient solutions to simplex range
searching exist, at least in the worst-case setting.} 
\fi
Lower bounds for this trade-off have been given by Chazelle \cite{c-lbcpr-89} under
the semigroup model of computation, where subtraction of the point weights is
not allowed. 
\iffull  It is possible that, say, the counting version of the 
simplex range searching problem, where we ask
just for the number of points in the query simplex,  might
admit better solutions using subtractions,  but no such solutions are known.
Moreover, there are recent lower-bound results when 
subtractions are also allowed; see \cite{La11} and references therein.
\fi

The data structures proposed for simplex range searching over the last two 
decades~\cite{Ma:ept,Ma-ehc} match the known lower bounds
within polylogarithmic factors. The state-of-the-art upper bounds are 
by
(i) Chan \cite{chan-opttrees}, who, building on many earlier results, 
provides a linear-size data structure with $O(n\log n)$
expected preprocessing time and $O(n^{1-1/d})$ query time, and
(ii) Matou\v{s}ek~\cite{Ma-ehc}, who provides a data structure with 
$O(n^d)$ storage, $O((\log n)^{d+1})$ query time, and
$O(n^d(\log n)^\eps)$ preprocessing time.\footnote{%
Here and in the sequel,
$\eps$ denotes an arbitrarily small positive constant. The implicit constants
in the asymptotic notation may depend on it, generally tending to infinity as
$\eps$ decreases to $0$.} 
A trade-off between space and query time can be obtained by combining 
these two data structures~\cite{Ma-ehc}.

Yao and Yao~\cite{YY} were perhaps the first to consider range searching 
in which ranges were delimited by graphs of polynomial functions.
Agarwal and Matou\v{s}ek~\cite{AM} have introduced a systematic study of
semialgebraic range searching. Building on the techniques developed
for simplex range searching, they presented a linear-size data structure with
$O(n^{1-1/b+\eps})$ query time, where $b=\max(d,2d-4)$.  For $d\le 4$, this
almost matches the performance for the simplex range searching, but for $d\ge
5$ there is a gap in the exponents of the corresponding bounds. 
Also see~\cite{ShSh} for related recent developments.

The bottleneck in the performance of the just mentioned 
range-searching data structure of \cite{AM} is
a combinatorial geometry problem, known as the \emph{decomposition of
arrangements into constant-complexity cells}. Here, we are given a set $\Sigma$
of $t$ algebraic surfaces in $\R^d$ (i.e., zero sets of $d$-variate polynomials),
with degrees bounded by a constant $\Delta_0$, and we want to decompose
each cell of the arrangement $\A(\Sigma)$ (see Section~\ref{sec:cross} for details)
into subcells that are constant-complexity semialgebraic sets, i.e., belong
to $\Gamma_{d,\Delta,s}$ for some constants $\Delta$ (bound on degrees)
and $s$ (number of defining polynomials), which may depend on $d$ 
and $\Delta_0$, but not on~$t$.  The crucial quantity is the total number 
of the resulting subcells over all cells of $\A(\Sigma)$; namely, if one can 
construct such a decomposition with $O(t^b)$ subcells, with some constant 
$b$, for every $t$ and $\Sigma$,
then the method of \cite{AM} yields query time $O(n^{1-1/b+\eps})$ 
(with linear storage).
The only known general-purpose technique for producing such a decomposition 
is the so-called \emph{vertical decomposition} \cite{CEGS,SA}, which 
decomposes $\A(\Sigma)$ into roughly $t^{2d-4}$ constant-complexity subcells, 
for $d\ge 4$~\cite{Kol,SA}.  
  
An alternative approach,
based on \emph{linearization}, was also proposed in~\cite{AM}. It maps the
semialgebraic ranges in $\R^d$ to simplices in some higher-dimensional space
and uses simplex range searching there. However, its performance depends on the
specific form of the polynomials defining the ranges. In some special cases
(e.g., when ranges are balls in $\reals^d$), linearization yields better query time 
than the decomposition-based technique mentioned above,
but for general constant-complexity semialgebraic ranges, 
linearization has worse performance.

\heading{Our results. } In a recent breakthrough, Guth and Katz~\cite{GK2}
have presented a new space decomposition technique, called polynomial
partitioning.  For a set $P\subset \R^d$  of $n$ points and a real parameter
$r$, $1<r\le n$, an \emph{$r$-partitioning polynomial} for $P$ is a nonzero
$d$-variate polynomial $f$ such that each connected component of $\R^d\setminus
Z(f)$ contains at most $n/r$ points of $P$, where $Z(f):=\{x\in\R^d \mid
f(x)=0\}$ denotes the zero set of~$f$. The decomposition of $\R^d$ into $Z(f)$
and the connected components of $\R^d\setminus Z(f)$ is called a
\emph{polynomial partition} (induced by $f$).  Guth and Katz show that an
$r$-partitioning polynomial of degree $O(r^{1/d})$ always exists, but 
their argument does not lead to an efficient algorithm for constructing 
such a polynomial, mainly because it relies on ham-sandwich cuts in 
high-dimensional spaces, for which no efficient construction is known.
Our first result is an efficient randomized algorithm for computing an
$r$-partitioning polynomial.

\begin{theorem}
\label{thm:partition-algo} 
Given a set $P$ of $n$ points in
$\R^d$, for some fixed $d$, and a parameter $r\le n$, an $r$-partitioning
polynomial for $P$ of degree $O(r^{1/d})$ can be computed in randomized
expected time $O(nr+r^3)$.  \end{theorem}

Next, we use this algorithm to bypass the arrangement-decomposition 
problem mentioned above. Namely, based on 
polynomial partitions, we construct \emph{partition trees}~\cite{AE,Ma:range}
that answer range 
queries with constant-complexity semialgebraic sets in near-optimal time,
using linear storage. An essential ingredient in the performance
analysis of these partition trees is a recent combinatorial
result of Barone and Basu \cite{BB}, originally conjectured
by the second author, which deals with the complexity of certain kinds of 
arrangements of zero sets of polynomials (see Theorem~\ref{t:basu}).
While there have already been several combinatorial applications of
the Guth-Katz technique (the most impressive being the original one in
\cite{GK2}, which solves the famous Erd{\H o}s's distinct distances problem,
and some of the others presented in~\cite{KMS,KMSS,SoTa,Za}), 
ours seems to be the first \emph{algorithmic} application.

We establish two range-searching results, both based on polynomial partitions.
%
For the first result, we need to introduce the notion of \emph{$D$-general
position}, for an integer $D\ge 1$. We say that a set $P\subset \R^d$ is in
$D$-general position if no $k$ points of $P$ are contained in the zero set of a
nonzero $d$-variate polynomial of degree at most $D$, where 
$k:={D+d\choose d}$. This is the number one expects for a ``generic'' point 
set.\footnote{Indeed, $d$-variate polynomials of degree at most $D$ have at most
$k-1$ distinct nonconstant monomials.
The Veronese map (e.g., see~\cite{GK2}) maps $\R^d$ to $\R^{k-1}$, and 
hyperplanes in $\R^{k-1}$ correspond bijectively to $k$-variate polynomials 
of degree at most $D$. It follows that any set of $k-1$ points in $\R^d$ is
contained in the zero set of a $d$-variate polynomial of degree at most $D$,
corresponding to the hyperplane in $\R^{k-1}$ passing through the Veronese 
images of these points. Similarly, $k$ points in general position are not 
expected to have this property, because one does not expect their images
to lie in a common hyperplane. 
 See \cite{EKS,GK1} for more details.}

\begin{theorem}\label{t:const-r} Let $d,\Delta,s$ and $\eps>0$ be constants.
Let
$P\subset \R^d$ be an $n$-point set in $D_0$-general position, where $D_0$ is a
suitable constant depending on $d,\Delta$, and $\eps$.  Then the
$\Gamma_{d,\Delta,s}$-range searching problem for 
$P$ can be solved with $O(n)$ storage,
$O(n\log n)$ expected preprocessing time, and $O(n^{1-1/d+\eps})$ query 
time.
\end{theorem}

\iffull
We note that both here and in 
the next theorem, while the preprocessing algorithm
is randomized, the queries are answered deterministically,
and the query time bound is worst-case.
\fi

Of course, we would like to handle arbitrary point sets, not only those in
$D_0$-general position. This can be achieved by an infinitesimal perturbation
of the points of $P$. 
A general technique known as ``simulation of simplicity'' 
(in the version considered by Yap \cite{y-gctsp-90})
ensures that the perturbed set $P'$ is in $D_0$-general position.
If a point $p\in P$ lies
in the interior of a query range $\gamma$, then so does the corresponding
perturbed point $p'\in P'$, and similarly for  $p$ in the interior of
$\R^d\setminus\gamma$.  However, for $p$ on
the boundary of $\gamma$, we cannot be sure if $p'$ ends up inside 
or outside~$\gamma$. 
 
Let us say that a \emph{boundary-fuzzy} solution to the 
$\Gamma_{d,\Delta,s}$-range searching problem is a data structure that, 
given a query $\gamma\in \Gamma_{d,\Delta,s}$,
returns an answer in which all points of $P$ in the interior of $\gamma$ are
counted and none in the interior of $\R^d\setminus\gamma$ is counted, while
each point $p\in P$ on the boundary of $\gamma$ may or may not be counted.  In
some applications, we can think of the points of $P$ being imprecise anyway
(e.g., their coordinates come from some imprecise measurement), and then
boundary-fuzzy range searching may be adequate.

\begin{corol}\label{c:fuzzy} Let $d,\Delta,s$, and $\eps>0$ be constants.
Then for every $n$-point set in $\R^d$, there
is a boundary-fuzzy $\Gamma_{d,\Delta,s}$-range searching data structure
with $O(n)$ storage, $O(n\log n)$ expected preprocessing time, and
$O(n^{1-1/d+\eps})$ query time.  \end{corol}
%
%
%

\iffull
Actually, previous results on range searching that use simulation of simplicity to avoid degenerate
cases also solve only the boundary-fuzzy variant (see e.g. \cite{Ma:ept,Ma-ehc}). 
However, the previous
techniques, even if presented only for point sets in general position, can
usually be adapted to handle degenerate cases as well, perhaps with some
effort, which is nevertheless routine. For our technique,
degeneracy appears to be a more substantial problem
because it is possible that a large subset of $P$
(maybe even all of $P$) is contained in the zero set of the partitioning
polynomial $f$, and the recursive divide-and-conquer mechanism yielded
by the partition of $f$ does not apply to this subset.

Partially in response to this issue, we 
\fi
present a different data structure that,
at a somewhat higher preprocessing cost, not only gets rid of the 
boundary-fuzziness condition  but also has a slightly improved 
query time (in terms of $n$). The main idea is that we build an 
auxiliary recursive data structure to handle 
the potentially large subset of points that lie in the 
zero set of the partitioning polynomial.

\begin{theorem}\label{t:large-r} 
Let $d,\Delta,s$, and $\eps>0$ be constants.
Then the $\Gamma_{d,\Delta,s}$-range searching problem for an arbitrary
$n$-point set in $\R^d$
can be solved with $O(n)$ storage, $O(n^{1+\eps})$ expected
preprocessing time, and $O(n^{1-1/d}\log^B n)$ query time, where 
$B$ is a constant depending on $d,\Delta,s$ and~$\eps$.  
\end{theorem}

We remark that the dependence of $B$ on $\Delta$, $s$, and $\eps$ is
reasonable, but its dependence on $d$ is superexponential.

Our algorithms work for the semigroup model described earlier.
Assuming that a semigroup operation can be executed in constant time, the query time remains 
the same as for the counting query.
A reporting query---report the points of $P$ lying in a query
range---also fits in the semigroup model, 
except one cannot assume that a semigroup operation in this case
takes constant time. The time taken by a reporting query is proportional to the 
cost of a counting query plus the number of reported points.

\iffull
\heading{Roadmap of the paper.} Our algorithm is based on the
polynomial partitioning technique by Guth and Katz, 
and we begin by briefly reviewing it 
in Section~\ref{sec:GK}. Next, in
Section~\ref{sec:algo}, we describe
the randomized algorithm for constructing such a
partitioning polynomial. Section~\ref{sec:cross} presents an algorithm 
for computing the cells of a polynomial partition that are crossed by 
a semialgebraic range, and discusses several related topics.  
Section~\ref{sec:range1} presents our first data structure, which is as in
Theorem~\ref{t:const-r}. Section~\ref{s:CAD} describes the method
for handling points lying on the zero set of the partitioning polynomial,
and Section~\ref{sec:range2} presents our second data 
structure. We conclude in 
Section~\ref{sec:concl} by mentioning a few open problems.
\fi

\section{Polynomial Partitions} \label{sec:GK}

In this section we briefly review the Guth-Katz technique for later use.
We begin by stating their result.

\begin{theorem}[Guth-Katz~\cite{GK2}]\label{t:gk} Given a set $P$ of $n$
points in $\R^d$ and a parameter $r\le n$, there exists an $r$-partitioning
polynomial for $P$ of degree at most $O(r^{1/d})$ (for $d$ fixed).
\end{theorem}

The degree in the theorem is asymptotically optimal in the worst case 
because the number of connected components of $\R^d\setminus Z(f)$ is
$O((\deg f)^d)$ for every polynomial $f$  (see, e.g., Warren
\cite[Theorem~2]{w-lbanm-68}).

\begin{proof}[Sketch of proof.]
The Guth-Katz proof uses the  
\emph{polynomial ham sandwich} theorem
of Stone and Tukey~\cite{ST}, which we state here in
a version for finite point sets: 
\emph{If $A_1,\ldots,A_k$ are finite sets in
$\R^d$ and $D$ is an integer satisfying ${D+d\choose d}-1\ge k$, then there
exists a nonzero polynomial $f$ of degree at most $D$ that simultaneously
bisects all the sets $A_i$.}  Here ``$f$ bisects $A_i$'' means that $f>0$ in at
most $\lfloor|A_i|/2\rfloor$ points of $A_i$ and $f<0$ in at most
$\lfloor|A_i|/2\rfloor$ points of~$A_i$; $f$ might vanish at any number
of the points of $A_i$, possibly even at all of them.

Guth and Katz inductively construct collections
$\PP_0,\PP_1,\ldots, \PP_m$ of subsets of $P$.
For $j=0,1,\ldots, m$, $\PP_j$ consists of at most
$2^j$ pairwise-disjoint subsets of $P$, each of size at most
$n/2^j$; the union of these sets does not have to
contain all points of $P$.

Initially, we have $\PP_0=\{P\}$.  The algorithm stops as soon as each subset in 
$\PP_m$ has at most $n/r$ points.  This implies that
$m \le \lceil \log_2 r\rceil$.  Having constructed $\PP_{j-1}$, we use 
the polynomial ham-sandwich theorem to construct a
polynomial $f_j$ that bisects each set of
$\PP_{j-1}$, with $\deg f_j =O(2^{j/d})$
(this is indeed an asymptotic upper bound for the smallest $D$
satisfying ${D+d\choose d}-1\ge 2^{j-1}$, assuming $d$ to be a constant).
For every subset $Q\in \PP_{j-1}$, let $Q^+= \{ q \in Q \mid f_j(q) > 0\}$ 
and $Q^-= \{ q \in Q \mid f_j(q) < 0\}$.  
We set $\PP_j :=\{Q^+,Q^- \mid Q \in \PP_{j-1}\}$;
empty subsets 
are not included in $\PP_j$. 

The desired $r$-partitioning polynomial for $P$ is then the
product $f:=f_1f_2\cdots f_m$.
We have 
$$\deg f=\sum_{j=1}^m\deg f_j = \sum_{j=1}^mO(2^{j/d}) =O(r^{1/d}).$$
By construction, the points of $P$ lying in a single connected
component of $\R^d\setminus Z(f)$ belong to a  single member of  $\PP_m$,
which implies that each connected component contains at most $n/r$ points 
of~$P$.
\end{proof}

\begin{proof}[Sketch of proof of the Stone--Tukey polynomial 
ham-sandwich theorem.]
We begin by observing that ${D+d\choose d}-1$ is the number of all nonconstant
monomials of degree at most $D$ in $d$ variables.  Thus, we fix a
collection $\MM$ of $k\le {D+d\choose d}-1$ such monomials. Let 
$\Phi\colon
\R^d\to\R^k$ be the corresponding \emph{Veronese map}, which  maps a point
$x=(x_1,\ldots,x_d)\in\R^d$ to the $k$-tuple of the values at
$(x_1,\ldots,x_d)$ of the monomials from $\MM$. For example, for $d=2$, $D=3$,
and $k=8\le {3+2\choose 2}-1$, we may use
$\Phi(x_1,x_2)=(x_1,x_2,x_1^2,x_1x_2,x_2^2,x_1^3,x_1^2x_2,x_1x_2^2)
\in \R^8$\iffull,
where $\MM$ is the set of the eight monomials appearing as components 
of $\Phi$\fi. 

Let $B_i:=\Phi(A_i)\subset\R^k$ \iffull
be the image of the given $A_i$ under this
Veronese map, for \fi$i=1,\ldots,k$.  By the standard \emph{ham-sandwich theorem}
(see, e.g., \cite{Ma:ham}), there exists a hyperplane $h$ in $\R^k$ that
simultaneously bisects all the $B_i$'s, in the sense that each open halfspace
bounded by $h$ contains at most half of the points of each of the sets $B_i$.
In a more algebraic language, there is a nonzero $k$-variate linear polynomial,
which we also call $h$, that bisects all the $B_i$'s, in the sense of
being positive on at most half of the points of each $B_i$, and
being negative on at most half of the points of each $B_i$.
Then $f:=h\circ\Phi$ is the desired $d$-variate polynomial of degree at
most $D$ bisecting all the~$A_i$'s.
\end{proof}

\section{Constructing a Partitioning Polynomial} \label{sec:algo}

In this section we present an efficient randomized algorithm that, given a point
set $P$ and a parameter $r < n$, constructs an  $r$-partitioning polynomial. 
The main difficulty in converting the above proof of the Guth-Katz 
partitioning theorem into an efficient
algorithm is the use of the 
ham-sandwich theorem in the possibly
high-dimensional space $\R^k$. A straightforward algorithm for
computing  ham-sandwich cuts in $\R^k$ inspects all
possible ways of splitting the input point sets by a hyperplane,
and has running time about $n^k$.  Compared to
this easy upper bound, the best known
ham-sandwich algorithms can save a factor of about $n$
\cite{LMS}, but this is insignificant in higher
dimensions.  A recent result  of Knauer, Tiwari, and Werner~\cite{KTW} shows
that a certain incremental variant of computing a ham-sandwich cut is
$W[1]$-hard (where the parameter is the dimension), and thus one perhaps should
not expect much better exact algorithms.

We observe that the exact bisection of each $A_i$ is not needed in the
Guth-Katz construction---it is sufficient to replace the Stone--Tukey
polynomial ham-sandwich theorem by a weaker result, as described below.

\heading{Constructing a well-dissecting polynomial.} We say that a 
polynomial $f$ is \emph{well-dissecting} for a point set $A$ if $f>0$ 
on at most $\frac78|A|$ points of $A$ and $f<0$ on at most $\frac78|A|$ 
points of~$A$. Given point sets $A_1,\ldots,A_k$ in $\R^d$ with $n$ 
points in total, we present a Las-Vegas algorithm for 
constructing a polynomial $f$ of degree $O(k^{1/d})$
that is well-dissecting for at least $\lceil k/2\rceil$ of
the $A_i$'s.

As in the above proof of the Stone--Tukey polynomial ham-sandwich theorem, let
$D$ be the smallest integer satisfying ${D+d\choose d}-1\ge k$. We fix a 
collection $\MM$ of $k$ distinct nonconstant monomials of degree 
at most $D$, and let
$\Phi$ be the corresponding Veronese map.  For each $i=1,2,\ldots,k$, we pick a
point $a_i\in A_i$ uniformly at random and compute $b_i:=\Phi(a_i)$.  Let $h$
be a hyperplane in $\R^k$ passing through $b_1,\ldots,b_k$, which can be found
by solving a system of linear equations, in $O(k^3)$ time.

If the points $b_1,\ldots,b_k$ are not affinely independent, then $h$ is not
determined uniquely (this is a technical nuisance, which the reader may want to
ignore on first reading).  In order to handle this case, we prepare in 
advance,
before picking the $a_i$'s, \emph{auxiliary} affinely independent points
$q_1,\ldots,q_k$ in $\R^k$, which are in general 
position with respect to $\Phi(A_1),\ldots,\Phi(A_k)$;
here we mean the ``ordinary'' general position,
i.e., no unnecessary affine dependences, that involve some of the
$q_i$'s and the other points, arise.  The points $q_i$ can be chosen at
random, say, uniformly in the unit cube;  with high probability, 
they have the desired general position property.  (If we do not want to 
assume the capability of choosing a random real number, we can pick the
$q_i$'s uniformly at random from a sufficiently large discrete set.)
If the dimension of the affine hull of
$b_1,\ldots,b_k$ is $k'<k-1$, we choose the hyperplane $h$ through 
$b_1,\ldots,b_k$ and $q_1,\ldots,q_{k-k'-1}$.  If $h$ is not 
unique, i.e., $q_1, \ldots q_{k-k'-1}$ are not affinely independent 
with respect to $b_1,\ldots b_k$, which we can detect while solving
the linear system, we restart the algorithm  by choosing $q_1,\ldots,q_k$
anew and then picking new $a_1,\ldots,a_k$.  
In this way, after a  constant expected number of iterations,
we obtain the
uniquely determined hyperplane $h$ through $b_1,\ldots,b_k$ and
$q_1,\ldots,q_{k-k'-1}$ as above, and we let $f=h\circ\Phi$ denote the corresponding
$d$-variate polynomial.  We refer to these steps as one \emph{trial} of the 
algorithm. For each $A_i$, we check whether $f$ is 
well-dissecting for $A_i$.  If $f$ is well-dissecting for only
fewer than $k/2$ sets, then we discard $f$ and perform another trial.

We now analyze the expected running time of the algorithm.  The intuition is that $f$ is
expected to well-dissect a significant fraction, say  at least half, 
of the sets $A_i$.  This intuition is reflected in the next lemma.
Let $X_i$ be the indicator variable of the
event: \textit{$A_i$ is \textbf{not} well-dissected by $f$}.

\begin{lemma}\label{l:EXi} 
For every $i=1,2,\ldots, k$, $\Ex[X_i] \le 1/4$.  
\end{lemma}
\begin{proof} 
Let us fix $i$ and the choices of 
$a_j$ (and thus of $b_j=\Phi(a_j)$) for all $j\ne i$. Let $k_0$ be the 
dimension of $F_0$, the affine hull of $\{b_j \mid j\ne i\}$. Then 
the resulting hyperplane $h$ passes through the $(k-2)$-flat $F$
spanned by $F_0$ and $q_1,\ldots,q_{k-k_0-2}$, irrespective of which 
point of $A_i$ is chosen. If $a_i$, the point chosen from $A_i$, is
such that $b_i = \Phi(a_i)$ lies on
$F_0$, then $h$ also passes through $q_{k-k_0-1}$.

Put $B_i:=\Phi(A_i)$, and let us project the configuration orthogonally to a
2-di\-mension\-al plane $\pi$ orthogonal to $F$. Then $F$ appears as a point
$F^*\in\pi$, and $B_i$ projects to a (multi)set $B_i^*$ in $\pi$. The
random hyperplane $h$ projects to a random line $h^*$ in $\pi$, 
whose choice can be interpreted as follows: pick
$b_i^*\in B_i^*$ uniformly at random; if $b_i^*\ne F^*$, then $h^*$ is the
unique line through $b_i^*$ and $F^*$; otherwise, when $b_i^*=F^*$, 
$h^*$ is the unique line through $F^*$ and $q_{k-k_0-1}^*$;
by construction, $q_{k-k_0-1}^*\ne F^*$.
The indicator variable $X_i$ is $1$
if and only if the resulting $h^*$ has more than $\frac78|B_i^*|$ points 
of $B_i^*$,
counted with multiplicity,
(strictly) on one side.

The special role of $q_{k-k_0-1}^*$ can be eliminated if we first
move the points of $B_i^*$ coinciding with $F^*$ to the point 
$q_{k-k_0-1}^*$, and then slightly perturb the points so as to
ensure that all points of $B_i^*$ are distinct and lie at distinct
directions from $F^*$; it is easy to see
that these transformations cannot decrease the probability of $X_i=1$. 
Finally, we note that the side of $h^*$ containing a point
$b^* \in B_i^*$ only depends on the direction of
the vector $\overrightarrow{F^*b^*}$, so we can also assume the points of 
$B_i^*$ to lie on the unit circle around $F^*$.  

Using (a simple instance of)
the standard planar ham-sandwich theorem, we partition $B_i^*$ into 
two subsets $L_i^*$ and $R_i^*$
of equal size by a line through the center $F^*$.
Then we bisect $L_i^*$  by a ray from $F^*$, and we do the same 
for $R_i^*$.  It is easily checked (see Figure~\ref{f:semi2-quarters})
that there always exist two of the resulting quarters, one of $L_i^*$ 
and one of $R_i^*$ (the ones whose union forms an angle $\le\pi$
between the two bisecting rays),
such that every line connecting $F^*$ with a point 
in either quarter contains at least $\frac14 |B_i^*|$ points of $B_i^*$ 
on each side.
Referring to these quarters as ``good'', 
we now take one of the bisecting rays, say that of $L_i^*$, and rotate
it about $F^*$ away from the good quarter of $L_i^*$. Each of the
first $\frac18 |B_i^*|$ points that the ray encounters has the property that 
the line supporting the ray has at least $\frac18 |B_i^*|$ points of 
$B_i^*$ on each side.  This implies that, for at least half of the 
points in each of the two remaining quarters, the line connecting 
$F^*$ to such a point has at least $\frac18 |B_i^*|$ points of $B_i^*$ 
on each side. Hence at most $\frac14 |B_i|$ points of $B_i$
can lead to a cut that is not well-dissecting for~$B_i$.

\labfig{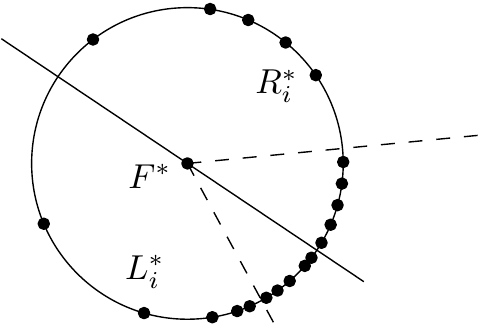}{Illustration to the proof of Lemma~\ref{l:EXi}. }

We conclude that, still conditioned on the choices of $a_j$, $j\ne i$, the
event $X_i=1$ has probability at most $1/4$. 
Since this holds for every choice
of the $a_j$, $j\ne i$, the unconditional probability of $X_i=1$ is also at
most $1/4$, and thus $\Ex[X_i]\le 1/4$ as claimed.  
\end{proof}

Hence, the expected number of sets $A_i$ that are not well-dissected by $f$ is
$$\Ex\bigl[ \sum_{i=1}^k X_i \bigr] = \sum_{i=1}^k \Ex[ X_i] \le k/4.$$ 
By Markov's inequality, with probability at least $1/2$, at least half of the
$A_i$'s are well-dissected by~$f$.  We thus obtain a polynomial that is
well-dissecting for at least half of the $A_i$'s after an expected constant
number of trials.

It remains to estimate the running time of each trial. The points $b_1,
\ldots, b_k$ can be chosen in $O(n)$ time.  Computing $h$ involves solving a
$k\times k$ linear system, which can be done in $O(k^3)$ time using Gaussian
elimination.
Note that we do
\emph{not} actually compute the entire sets $\Phi(A_i)$. No
computation is needed for passing from $h$ to $f$---we just re-interpret 
the
coefficients. To check which of $A_1, \ldots A_k$ are well-dissected by $f$,
we evaluate $f$ at each point of $A = \bigcup_i A_i$.
First we evaluate each of the $k$ monomials in $\MM$ at each point of $A$. If
we proceed incrementally, from lower degrees to higher ones, this can be done
with $O(1)$ operations per monomial and point of $A$, in $O(nk)$ time in
total. Then, in additional $O(nk)$ time, we compute the values of $f(q)$, for
all $q \in A$, from the values of the monomials.  Putting everything together
we obtain the following lemma.

\begin{lemma}\label{l:dissect} Given point sets $A_1,\ldots,A_k$ in $\R^d$ (for
fixed $d$) with $n$ points in total, a polynomial $f$ of degree $O(k^{1/d})$
that is well-dissecting for at least $\lceil k/2\rceil$ of the $A_i$'s can be
constructed in $O(nk+k^{3})$ randomized expected time.  \end{lemma}

\iffull
\begin{proof}[Constructing a partitioning polynomial: Proof of 
Theorem~\ref{thm:partition-algo}.] 
\else
\heading{Constructing a partitioning polynomial.}
\fi
We now describe the
algorithm for computing an $r$-partitioning polynomial $f$.  We essentially imitate
the Guth--Katz construction, with Lemma~\ref{l:dissect} replacing the polynomial
ham-sandwich theorem, but with an additional twist.

The algorithm works in phases.  At the end of the $j$-th phase, 
for $j \ge 1$, we have a family 
$f_1, \ldots, f_j$ of $j$ polynomials and a family $\PP_j$ 
of at most $2^j$ pairwise-disjoint
subsets of $P$, each of size at most $(7/8)^jn$.
Similar to the Guth--Katz construction, 
$\PP_j$ is not necessarily a partition of $P$, since 
the points of $P \cap Z(f_1f_2\cdots f_j)$ do not belong to
 $\bigcup \P_j$.  
Initially, $\PP_0 = \{P\}$.  
The algorithm stops when each set in $\PP_j$ has at most 
$n/r$ points. In the $j$-th phase, the algorithm constructs $f_j$ and
$\PP_j$ from $f_1,\ldots,f_{j-1}$ and $\PP_{j-1}$, as follows.

At the beginning of the $j$-th phase, let 
$\LL_j = \{ Q \in \PP_{j-1} \mid |Q| > (7/8)^jn \}$ be the family of the
``large'' sets in $\P_{j-1}$, and set
$\kappa_j = |\LL_j| \le (8/7)^j$. We also initialize the collection
$\P_j$ to $\P_{j-1}\setminus \LL_j$, the family of ``small''
sets in $\P_{j-1}$. Then we perform
at most $\lceil \log_2 \kappa_j\rceil$ dissecting steps, as follows:
After $s$ steps, we have a family 
$g_1, \ldots, g_s$ of polynomials, the current set $\P_j$, and a 
subfamily $\LL_j^{(s)} \subseteq \LL_j$ of size at most
$\kappa_j/2^s$,  consisting of the members of $\LL_j$
that were not well-dissected by any of $g_1, \ldots, g_s$.
If $\LL_j^{(s)} \ne \emptyset$ we choose, using Lemma~\ref{l:dissect}, 
a polynomial $g_{s+1}$ of degree at most $c(\kappa_j/2^s)^{1/d}$
(with a suitable constant $c$ that depends only on $d$) that 
well-dissects at least half of the members of $\LL_j^{(s)}$. 
For each $Q \in \LL_j^{(s)}$, let
$Q^+ = \{ q \in Q \mid g_{s+1}(q) > 0 \}$ and
$Q^- = \{ q \in Q \mid g_{s+1}(q) < 0 \}$. 
If $Q$ is well-dissected, i.e., $|Q^+|, |Q^-| \le \frac78|Q|$,
then we add $Q^+,Q^-$ to $\P_j$, and otherwise, we add $Q$ to 
$\LL_j^{(s+1)}$.  Note that in the former case the points $q\in Q$
satisfying $g_{s+1}(q)=0$ are ``lost'' and do not participate in the 
subsequent dissections.  By Lemma~\ref{l:dissect}, 
$|\LL_j^{(s+1)}| \le |\LL_j^{(s)}|/2 \le \kappa_j/2^{s+1}$. 

The $j$-th phase is completed when $\LL_j^{(s)} = \emptyset$, in which case
we set\footnote{%
  Note that $f_j$ is not necessarily well-dissecting, because it does not
  control the sizes of subsets with positive or with negative signs.}
$f_j := \prod_{\ell=1}^s g_\ell$.
By construction, each point set in $\P_j$ has 
at most $(7/8)^jn$ points, and the points of $P$ not belonging
to any set of  $\P_j$
lie in $Z(f_1\cdots f_j)$. Furthermore,
$$ 
\deg f_j \le \sum_{s\ge 0} c(\kappa_j/2^s)^{1/d} =O(\kappa_j^{1/d}) ,
$$
where again the constant of proportionality depends only on $d$.
Since every set in $\PP_{j-1}$ is split into at most two sets before being added to 
$\PP_j$, 
$|\PP_j| \le 2|\PP_{j-1}| \le 2^j$.

If $\PP_j$ contains subsets with more than $n/r$ points, 
we begin the $(j+1)$-st  phase with the current $\P_j$;
otherwise the algorithm stops and returns
$f := f_1f_2\cdots f_j$.
This completes the description of the algorithm.

Clearly, $m$, the number of phases of the algorithm, is at most
$\lceil \log_{8/7} r\rceil$.
Following the same argument as in~\cite{GK2},
and as briefly sketched in Section~\ref{sec:GK}, it can be shown
that all points lying in a single connected component of 
$\R^d\setminus Z(f)$ belong to a single member of $\P_m$, and thus
each connected component contains at most $n/r$ points of $P$.
Since the degree of $f_j$ is $O(\kappa_j^{1/d})$,
$\kappa_j \le (8/7)^j$, and $m \le \lceil \log_{8/7} r\rceil$,
we conclude that 
$$
\deg f = O\biggl(\sum_{j=1}^m \kappa_j^{1/d}\biggr) = 
	O\biggl(\sum_{j=1}^m (8/7)^{j/d} \biggr)= O(r^{1/d}).$$

As for the expected running time of the algorithm,
the $s$-th step of the $j$-th phase takes $O(n\kappa_j/2^s+ (\kappa_j/2^s)^3)$ 
expected time, so the $j$-th phase takes a total of 
$O(n\kappa_j+\kappa_j^3)$ expected 
time. Substituting 
$\kappa_j \le (8/7)^j$ in the above bound and summing over all $j$, the overall
expected running time of the algorithm is $O(nr+r^3)$. This completes the 
proof of Theorem~\ref{thm:partition-algo}.
\iffull
\end{proof}
\fi

\heading{Remark.} 
Theorem~\ref{thm:partition-algo}
is employed for the preprocessing in our range-searching algorithms
in Theorems~\ref{t:const-r} and~\ref{t:large-r}. In Theorem~\ref{t:const-r}
we take $r$ to be a large constant, and the expected running time
in Theorem~\ref{thm:partition-algo} is $O(n)$. 
However, in Theorem~\ref{t:large-r}, we require $r$ to be a small
fractional power of $n$, say $r=n^{0.001}$. 
It is a challenging open problem to improve the expected running time in 
Theorem~\ref{thm:partition-algo} to $O(n\polylog(n))$ 
when $r$ is such a small fractional power of $n$.
The bottleneck in the current algorithm 
is the subproblem of evaluating a given
$d$-variate polynomial $f$ of degree $D=O(r^{1/d})$ at $n$ given points;
everything else can be performed in $O(n\polylog (r)+r^{O(1)})$ expected 
time. Finding the signs of $f$ at those points would
actually suffice, but this probably does not make the problem
any simpler.

This problem of \emph{multi-evaluation} of multivariate real polynomials
has been considered in the literature, and there is a nontrivial
improvement over the straightforward $O(nr)$ algorithm,
due to N\"usken and Ziegler \cite{NZ}. 
\iffull
Concretely, in the bivariate case ($d=2$), their algorithm can evaluate
a bivariate polynomial of degree $D\le \sqrt n$ at $n$ given points
using $O(nD^{0.667})$ arithmetic operations. 
It is based on fast matrix multiplication,
and even under the most optimistic 
possible assumption on the speed of matrix multiplication,
it cannot get below $nD^{1/2}$. Although this is significantly faster
than our naive $O(nr)$-time algorithm, which is $O(nD^2)$
in this bivariate case, it is still a far cry
from what we are aiming at.
\fi
Let us remark that in a different setting, for polynomials
over finite fields (and over certain more general finite rings), there
is a remarkable method for multi-evaluation by Kedlaya 
and Umans~\cite{KU} achieving $O(((n+D^d)\log q)^{1+\eps})$ running time,
where $q$ is the cardinality of the field.

\section{Crossing a Polynomial Partition with a Range}
\label{sec:cross}

In this section we define the crossing number of  a polynomial partition 
and describe an algorithm for computing the cells of a polynomial 
partition that are crossed by a semialgebraic range, both of which 
will be crucial for our range-searching data structures. 
We begin by recalling a few results on arrangements of algebraic surfaces.
We refer the reader to \cite{SA} for a comprehensive review of such arrangements.

Let $\Sigma$ be a finite set of algebraic surfaces in $\reals^d$. The 
\emph{arrangement} of $\Sigma$, denoted by $\A(\Sigma)$, is the partition
of $\reals^d$ into maximal relatively open connected subsets, 
called \emph{cells}, such that all points within each cell lie in the 
same subset of surfaces of $\Sigma$ (and in no other surface).
If $\F$ is a set of $d$-variate polynomials, then with a slight abuse of notation,
we use $\A(\F)$ to denote the arrangement $\A(\{Z(f)\mid f\in\F\})$
of their zero sets.
We need the following result on arrangements, which follows from Proposition~7.33 and Theorem~16.18 in~\cite{BPR}.

\begin{theorem}[{Basu, Pollack and Roy~\cite{BPR}}]
\label{t:make-arrg}
Let $\F=\{f_1,\ldots,f_s\}$ be a set of $s$ real 
$d$-variate polynomials, each of degree
at most~$\Delta$. Then the arrangement $\A(\F)$ in $\R^d$
has at most 
$O(1)^d(s\Delta)^{d}$ 
cells, and it can be computed
in time at most $T= s^{d+1} \Delta^{O(d^4)}$.
Each cell is described as a semialgebraic set
using at most $T$ polynomials of degree bounded by
$\Delta^{O(d^3)}$. Moreover, the algorithm 
supplies an explicitly computed point in each cell.
\end{theorem}

A key ingredient for the analysis of our range-searching
data structure is the following recent result of Barone and Basu~\cite{BB},
which is a refinement of a series of previous studies; 
e.g., 
\iffull
see~\cite{bpr-ncdfp,BPR}: 
\fi
\begin{theorem}[{Barone and Basu~\cite{BB}}]
\label{t:basu}
Let $V$ be a $k$-dimensional  algebraic variety in $\R^d$
defined by a finite set $\G$  of $d$-variate polynomials, each of degree at most 
$\Delta$, and let $\F$ be a set of $s$ polynomials of degree at most 
$D\ge\Delta$.  Then  the number of cells
of $\A(\F\cup\G)$ (of all dimensions) that are contained in $V$
is bounded by $O(1)^d \Delta^{d-k}(sD)^k$.
\end{theorem}

\heading{The crossing number of polynomial partitions. }
%
Let $P$ be a set of $n$ points in $\R^d$,
and let $f$ be an $r$-partitioning polynomial for~$P$.
Recall that the \emph{polynomial partition} 
$\thePartit=\thePartit(f)$
induced by $f$ is the partition of $\R^d$ into
the zero set $Z(f)$ and the connected components
$\omega_1,\omega_2,\ldots,\omega_t$ of $\R^d\setminus Z(f)$. 
As already noted, Warren's theorem~\cite{w-lbanm-68}
implies that $t=O(r)$.
We call $\omega_1,\ldots,\omega_t$ the \emph{cells} of $\thePartit$
(although they need not be cells in the sense typical,
e.g., in topology; they need not even be simply connected).
$\thePartit$ also induces a partition $P^*, P_1,\ldots, P_t$ of $P$,
where $P^*=P\cap Z(f)$ is the \emph{exceptional part},
and $P_i=P\cap \omega_i$, for $i=1,\ldots,t$, are the \emph{regular parts}.
By construction, $|P_i|\le n/r$ for every $i=1,2,\ldots,t$,
but we have no control over the size of $P^*$---this 
will be the source of most of our technical difficulties.

Next, let $\gamma$ be a range in ${\Gamma}_{d,\Delta,s}$.
We say that $\gamma$ \emph{crosses} a cell $\omega_i$ if 
neither $\omega_i\subseteq \gamma$ nor $\omega_i\cap\gamma=\emptyset$.
The \emph{crossing number} of $\gamma$ is the number of 
cells of $\thePartit$ crossed by $\gamma$, and the \emph{crossing number}
of $\thePartit$ (with respect to~${\Gamma}_{d,\Delta,s}$)
is the maximum of the crossing numbers of all 
$\gamma\in {\Gamma}_{d,\Delta,s}$.
Similar to many previous range-searching algorithms~\cite{chan-opttrees,Ma:ept,Ma-ehc},
the crossing number of $\Omega$ will determine the query time of our 
range-searching algorithms described in 
Sections~\ref{sec:range1} and~\ref{sec:range2}.

\begin{lemma}\label{l:cr} 
If $\thePartit$ is a polynomial partition induced by an $r$-partitioning 
polynomial of degree at most $D$, then the crossing number of 
$\thePartit$ with respect to~${\Gamma}_{d,\Delta,s}$, with $\Delta\le D$,
is at most $C s\Delta D^{d-1}$, where $C$ is a suitable 
constant depending only on~$d$.
\end{lemma}

\begin{proof}
Let $\gamma\in {\Gamma}_{d,\Delta,s}$; then 
$\gamma$ is a Boolean combination of
up to $s$ sets of the form $\gamma_j:=\{x\in\R^d \mid g_j(x)\ge 0\}$,
where $g_1,\ldots,g_s$ are polynomials of degree at most
$\Delta$. If $\gamma$ crosses a cell $\omega_i$, 
then at least one of the ranges $\gamma_j$ also crosses $\omega_i$, 
and thus it suffices to establish that the crossing number 
of any range $\gamma$, defined by a single $d$-variate
polynomial inequality $g(x)\ge 0$ of degree at most $\Delta$,
is at most $C\Delta D^{d-1}$.

We apply Theorem~\ref{t:basu} with $V:=Z(g)$, which is an algebraic variety
of dimension $k\le d-1$, and with $s=1$ and $\F=\{f\}$, where $f$ is
the $r$-partitioning polynomial. Then, for each cell $\omega_i$ crossed
by $\gamma$, $\omega_i\cap Z(g)$ is a nonempty
union of some of the cells in 
$\A(\F\cup\{g\}) = \A(\{f,g\})$ that lie in $V$. Thus, the crossing number of 
$\gamma$ is
at most $O(1)^d\Delta D^{d-1}$.
\end{proof}

\heading{Algorithmic issues. }
We need to perform the following algorithmic primitives (for $d$ \emph{fixed}
as usual) for the range-searching algorithms that we will later present:

\begin{itemize}
\item[{\bf (A1)}] 
Given an $r$-partitioning polynomial $f$
of degree $D=O(r^{1/d})$, compute (a suitable representation of)
the partition $\thePartit$ and the induced partition
of $P$ into $P^*,P_1,\ldots,P_t$. 

By computing  $\A(\{f\})$, using Theorem~\ref{t:make-arrg}, and then 
testing the membership of each point $p\in P$ in each cell $\omega_i$
in time polynomial in~$r$, the above operation can be performed in
$O(nr^{c})$ time,\iffull\footnote{%
Of course, this is somewhat inefficient,
and it would be nice to have a fast point-location algorithm
for the partition $\thePartit$---this would be the second step, together
with an improved construction of an $r$-partitioning polynomial $f$
(concretely, an improved multi-point evaluation procedure for $f$)
as discussed at the end of Section~\ref{sec:algo}, needed to improve
the preprocessing time in Theorem~\ref{t:large-r}.}
\fi
where $c=d^{O(1)}$.

\item[{\bf (A2)}]
Given (a suitable representation of)
$\thePartit$ as in (A1) and a query range 
$\gamma\in{\Gamma}_{d,\Delta,1}$, i.e., a range defined
by a single $d$-variate polynomial $g$ of degree $\Delta\le D$, 
compute which of the cells of $\thePartit$ are crossed by $\gamma$ and which are
completely contained in $\gamma$.

We already have the arrangement $\A(\{f\})$, and we compute
$\A(\{f,g\})$. For each cell of  $\A(\{f,g\})$ contained
in $Z(g)$, we locate its representative point in $\A(\{f\})$,
and this gives us the cells crossed by $\gamma$. For the remaining
cells, we want to know whether they are inside $\gamma$ or outside,
and for that, it suffices to determine the sign of $g$ at the representative
points.
Using Theorem~\ref{t:make-arrg}, the above task  can thus be accomplished
in time $O(r^c)$, with $c=d^{O(1)}$.
\end{itemize}

\section{Constant Fan-Out Partition Tree}
\label{sec:range1}

We are now ready to describe our first data structure for 
${\Gamma}_{d,\Delta,s}$-range searching, which is 
 a constant fan-out (branching degree) partition tree,
and which works for points in general position.

\iffull
\begin{proof}[Proof of Theorem~\ref{t:const-r}.]
\fi
Let $P$ be a set of $n$ points in $\R^d$, and let $\Delta, s$ 
be constants. 
We choose $r$ as a (large) constant depending on $d,\Delta,s$, and 
the prespecified parameter $\eps$. 
We assume $P$ to be in $D_0$-general position for some sufficiently large 
constant $D_0 \gg r^{1/d}$.  We construct a partition tree $\T$ of 
fan-out $O(r)$ as follows.  We first construct an $r$-partitioning
polynomial $f$ for $P$ using Theorem~\ref{thm:partition-algo}, and compute the 
partition $\thePartit$ of $\R^d$ induced by $f$, as well as the corresponding 
partition $P=P^*\cup P_1\cup\cdots\cup P_t$ of $P$, where $t=O(r)$. 
Since $r$ is a constant, the (A1) operation, discussed in
Section~\ref{sec:cross}, performs this computation in $O(n)$ time. 
We choose $D_0$ so as to ensure that it is at least $\deg f$, 
and then our assumption that $P$ is in $D_0$-general position implies that
the size of $P^* = P \cap Z(f)$ is bounded by $D_0$.

We set up the root of $\T$, where we store 
\iffull
\begin{enumerate}
\item[(i)] the partitioning polynomial $f$, and a suitable representation
of the partition $\thePartit$;
\item[(ii)] a list of the points of the exceptional part $P^*$; and 
\item[(iii)] $w(P_i)$, the sum of the weights of the points of the regular part $P_i$, 
for each $i=1,2,\ldots,t$.
\end{enumerate}
\fi
The regular parts $P_i$ are not stored explicitly at the root.
Instead, for each $P_i$ we recursively build a subtree representing it.
The recursion terminates, at leaves of $\T$, as soon as we reach point sets 
of size smaller than a suitable constant $n_0$. The points of each such
set are stored explicitly at the corresponding leaf of $\T$.

Since each node of $\T$ requires only a constant amount
of storage and each point of $P$ is stored at only one node of $\T$, 
the total size of $\T$ is $O(n)$. The preprocessing time
is $O(n\log n)$ since $\T$ has depth $O(\log_r n)$
and each level is processed in $O(n)$ time.

To process a query range $\gamma\in {\Gamma}_{d,\Delta,s}$,
we start at the root of $\T$ and maintain a global counter
which is initially set to $0$.
Among the cells $\omega_1,\ldots,\omega_t$ of the 
partition $\thePartit$ stored at the root, we find, using the (A2)
operation, those completely contained in $\gamma$, and those 
crossed by $\gamma$. Actually, we compute a superset
of the cells that $\gamma$ crosses, namely, the cells crossed
by the zero set of at least one of the (at most $s$) 
polynomials defining $\gamma$.
For each cell $\omega_i\subseteq \gamma$, we add
the weight $w(P_i)$ to the global counter.
We also add to the global counter the weights of the
points in $P^*\cap\gamma$, which we find by testing each
point of $P^*$ separately. Then we recurse in each subtree
corresponding to a cell $\omega_i$ crossed by $\gamma$ 
(in the above weaker sense).  
The leaves, with point sets of size $O(1)$, are processed 
by inspecting their points individually. 
By Lemma~\ref{l:cr}, the number of cells crossed by any of the polynomials
defining $\gamma$ 
at any interior node of $\T$ is at most 
$Cs\Delta D^{d-1}\le C'r^{1-1/d}$,
where $C'=C'(d,s,\Delta)$ is a constant independent of~$r$.

The query time $\qtime(n)$ obeys the 
following recurrence:
$$
\qtime(n)\le \left \{ 
	\begin{array}{ll}
	 C'r^{1-1/d} \qtime(n/r) + O(1) & \mbox{for $n > n_0$,} \\[1mm]
	O(n) & \mbox{for $n \le n_0$,}
	\end{array} \right .
$$
It is well known (e.g., see~\cite{Ma:ept}), and easy to check, that the recurrence solves
to $\qtime(n)=O(n^{1-1/d+\eps})$, for every fixed $\eps>0$, with an
appropriate sufficiently large choice of $r$ 
as a function of $C'$ and $\eps$, and with an appropriate choice of $n_0$. 
This concludes the proof of Theorem~\ref{t:const-r}. 
\iffull
\end{proof}
\fi

\iffull
\begin{proof}[Boundary-fuzzy range searching: 
Proof of Corollary~\ref{c:fuzzy}.]
Now we consider
 the case where the points of $P$ are not necessarily in $D_0$-general position.
As was mentioned in the introduction, we apply a general perturbation 
scheme of Yap \cite{y-gctsp-90} to the previous range-searching algorithm.

Yap's scheme is applicable to an algorithm whose input is
a sequence of real numbers (in our case, the $dn$ point coordinates
plus the coefficients in the polynomials specifying the query range).
It is assumed that the algorithm makes decision steps by way of
evaluating polynomials with rational coefficients
taken from a finite set $\PP$, where the input
parameters are substituted for the variables. The algorithm makes
a 3-way branching depending on the sign of the evaluation.
The set $\PP$ does not depend on the input. 
The input is considered
degenerate if one of the signs in the tests is~0. 

Yap's scheme provides a black box for evaluating the polynomials from $\PP$
that, whenever the actual value is $0$, 
also supplies a nonzero sign, $+1$ or $-1$,
which the algorithm may use for the branching, instead of the zero sign.
Thus, the algorithm never ``sees'' any degeneracy.
Yap's method guarantees that these signs are consistent, i.e.,
for every input, the branching done in this way corresponds to some
infinitesimal perturbation of the input sequence, and so does the
output of the algorithm (in our case, the answer to a range-searching
query).

For us, it is important that if the degrees of the polynomials in $\PP$
are bounded by a constant, the black box also operates in time bounded by
a constant (which is apparent from the explicit specification
in  \cite{y-gctsp-90}). Thus, applying the perturbation scheme 
influences the running time only by a multiplicative constant.

It can be checked the range-searching
algorithm presented above is of the required kind, with all
branching steps based on the sign of suitable polynomials 
in the coordinates of the input points and in the coefficients of the 
polynomials in the query range, and the degrees of these polynomials 
are bounded by a constant. For producing the
partitioning polynomial $f$, we solve systems of linear equations,
and thus the coefficients of $f$ are given by certain determinants
obtained from Cramer's rule. The computation of the polynomial partition
and locating points in it is also based on the signs of suitable
bounded-degree polynomials, as can be checked by inspecting the
relevant algoritms, and similarly for intersecting a polynomial
partition with the query range. The key fact is that all computations
in the algorithm are of constant-bounded depth---each of the values
ever computed is obtained from the input parameters by a constant number
of arithmetic operations.

We also observe that when Yap's scheme is applied, the algorithm
never finds more than $D_0$ points in the exceptional set $P^*$
(in any of the nodes of the partition tree). Indeed, if $D_0+1$ input points
lie in the zero set of a polynomial $f$ as in the algorithm,
then a certain polynomial in the coordinates of these $D_0+1$ points
vanishes (see, e.g., \cite[Lemma~6.3]{KMS}). Thus, assuming that
the algorithm found $D_0+1$ points on $Z(f)$, it could test
the sign of this polynomial at such points, and the black box would return
a nonzero sign, which would contradict the consistency of
Yap's scheme. 

After applying Yap's scheme,  the
preprocessing cost, storage, and query time remain asymptotically 
the same as in Theorem~\ref{t:const-r} (but with
larger constants of proportionality). 
Since the output of the
algorithm corresponds to some infinitesimally perturbed
version of the input (point set and query range),
we obtain a boundary-fuzzy answer for the original point set. 
\end{proof}
\fi

\section{Decomposing a Surface into Monotone Patches}
\label{s:CAD}

As mentioned in the Introduction, 
if we construct an $r$-partitioning polynomial $f$
for an arbitrary point set $P$, the exceptional set $P^*=P\cap Z(f)$ may 
be large, as is schematically indicated in
Fig.~\ref{f:semi2-patches}~(left). Since $P^*$ is not partitioned
by $f$ in any reasonable sense, it must be handled differently, as described below.

\labfig{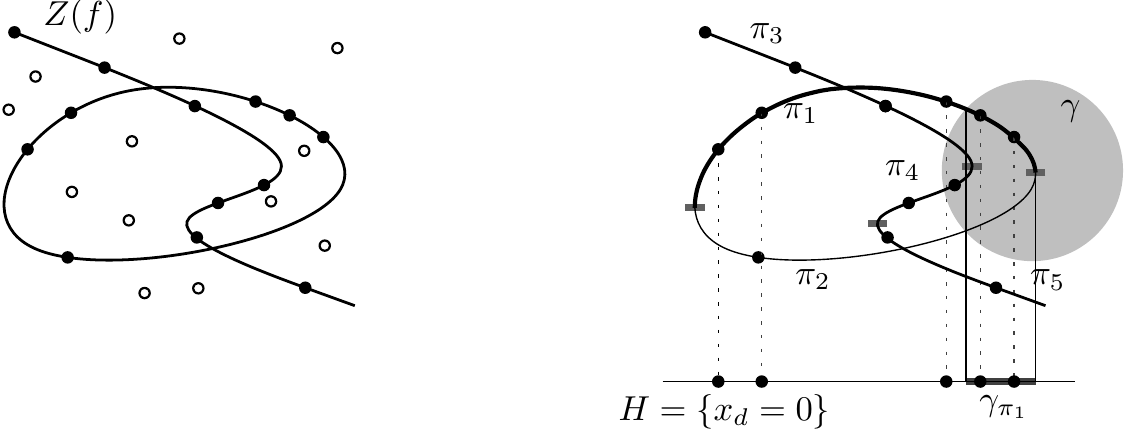}{The zero set of the partitioning polynomial (left), and its decomposition into monotone patches that project to the hyperplane $H$ bijectively. Only the $1$-dimensional patches are labeled.}

Following the terminology in~\cite{Hi,SS2}, we call a direction 
$v \in\sphere^{d-1}$ \emph{good} for $f$ if, for every $a\in\reals^d$, the
polynomial $p(t)=f(a+vt)$ does not vanish identically; that is,
any line in direction $v$ intersects $Z(f)$ at finitely many points. As 
argued in~\cite[pp.~304--305 and pp.~314--315]{SS2}, a random direction is good for $f$ with probability~1. 
By choosing a good direction and rotating the coordinate system, 
we assume that the $x_d$-direction, referred to
as the \emph{vertical} direction, is good for $f$.

In order to deal with $P^*$, we partition $Z(f)$ into finitely many pieces,
called \emph{patches}, in such a way that each of the patches is
\emph{monotone} in the vertical direction, meaning that every line 
parallel to the $x_d$-axis
intersects it at most once. This is illustrated,
in the somewhat trivial 2-dimensional setting,
in Fig.~\ref{f:semi2-patches}
\iffull
(right):
\fi
 there are five one-dimensional patches $\pi_1,\ldots,\pi_5$, plus
four 0-dimensional patches. 
Then we treat each patch $\pi$ separately:
We project the point set $P^*\cap \pi$ orthogonally to the coordinate
hyperplane $H:=\{x_d=0\}$, and we preprocess the projected set,
denoted $P^*_\pi$,
for range searching with suitable ranges. These ranges are projections
of ranges of the form $\gamma\cap\pi$, where $\gamma\in\Gamma_{d,\Delta,s}$
is one of the original ranges. In Fig.~\ref{f:semi2-patches}~
\iffull
(right),
\fi
the patch 
$\pi_1$ is drawn thick, a range $\gamma$ is depicted as a gray disk,
and the projection $\gamma_{\pi_1}$ of $\gamma\cap\pi_1$ is shown as
a thick segment in~$H$.

The projected range $\gamma_\pi$ is typically more complicated
than the original range $\gamma$ (it involves more polynomials of larger degrees),
but, crucially, it is only $(d-1)$-dimensional, and 
$(d-1)$-dimensional queries can be processed somewhat more
efficiently than $d$-dimensional ones, which makes the
whole scheme work. We will discuss this in more detail
in Section~\ref{sec:range2} below, but first we recall the notion of
\emph{cylindrical algebraic decomposition} (CAD,
or also \emph{Collins decomposition}), which is
a tool that allows us to decompose $Z(f)$ into monotone patches,
and also to compute the projected ranges~$\gamma_\pi$.

Given a finite set $\F=\{f_1,\ldots,f_s\}$
of $d$-variate polynomials, a \emph{cylindrical
algebraic decomposition adapted to $\F$}  is a way of decomposing
$\R^d$ into a finite collection of relatively open \emph{cells},
which have a simple shape (in a suitable sense),
 and which refine the arrangement $\A(\F)$.
We refer, e.g., to \cite[Chap.~5.12]{BPR}
for the definition and construction of the ``standard'' CAD.
Here we will use a simplified variant, which can
be regarded as the ``first stage'' of the standard CAD, and
which is captured by \cite[Theorem~5.14, Algorithm~12.1]{BPR}.
We also refer to \cite[Appendix~A]{SS2} for a concise treatment,
which is perhaps more accessible at first encounter.

Let $\F$ be as above.
To obtain the first-stage CAD for $f$, one constructs 
a suitable collection $\E=\E(\F)$ of polynomials 
in the variables $x_1,\ldots,x_{d-1}$
(denoted by ${\rm Elim}_{X_k}(\F)$ in \cite{BPR}). Roughly speaking,
the zero sets of the polynomials in $\E$, viewed as subsets of 
the coordinate hyperplane $H$ (which is identified with $\R^{d-1}$),
contain the projection onto $H$ of all intersections $Z(f_i)\cap Z(f_j)$,
$1\le i<j\le s$, as well as the projection of the loci in $Z(f_i)$
where $Z(f_i)$ has a vertical tangent hyperplane, or a singularity
of some kind.  The actual construction of $\E$ is somewhat more complicated,
and we refer to the aforementioned references 
for more details.

\labfig{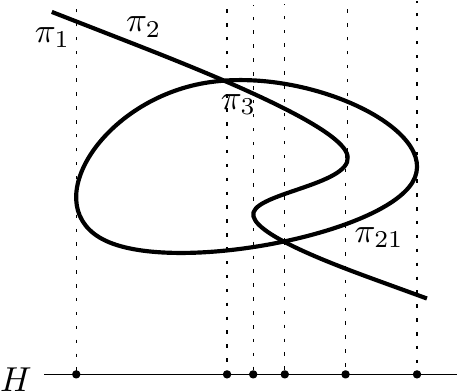}{A schematic illustration of the first-stage cylindrical algebraic decomposition.}

Having constructed $\E$, the first-stage CAD is obtained
as the arrangement $\A(\F\cup\E)$ in $\R^d$, where the polynomials
in $\E$ are now considered as $d$-variate polynomials (in which
the variable $x_d$ is not present). In geometric terms, 
we erect a ``vertical wall'' in $\R^d$
over each zero set within $H$ of a $(d-1)$-variate polynomial
from $\E$, and the CAD is the arrangement
of these vertical walls plus the zero sets of $f_1,\ldots,f_s$.
The first-stage CAD is illustrated in 
Fig.~\ref{f:semi2-CAD},
for the same (single) polynomial as in 
Fig.~\ref{f:semi2-patches}~(left). 

In our algorithm, we are interested in the cells of the CAD
that are contained in some of the sero sets $Z(f_i)$; these are going to be the
monotone patches alluded to above.  
\iffull
We note that
using the first-stage CAD for the purpose of decomposing
$Z(f)$ into monotone patches seems somewhat wasteful.
For example, the number of patches in Fig.~\ref{f:semi2-patches}
is considerably smaller than the number of patches in the CAD in
Fig.~\ref{f:semi2-CAD}.
But the CAD is simple and well known, and
(as will follow from the analysis in Section~\ref{sec:range2})
possible improvements in the number of patches (e.g.\ using the vertical-decomposition technique~\cite{SA}) do not seem
to influence our asymptotic bounds
on the performance of the resulting range-searching data structure.
\else
\fi
The following lemma summarizes the properties of the first-stage CAD
that we will need; we refer to \cite[Theorem~5.14, Algorithm~12.1]{BPR}
for a proof.
\begin{lemma}[Single-stage CAD]\label{l:CAD}
Given a set $\F=\{f_1,\ldots,f_s\}\subset \R[x_1,\ldots,x_d]$
of polynomials,
each of degree at most~$D$, there is a set $\E=\E(\F)$ of 
$O(s^2 D^3)$ polynomials in $\R[x_1,\ldots,x_{d-1}]$,
each of degree $O(D^2)$, which can
be computed in time $s^2 D^{O(d)}$, such that the first-stage
CAD defined by these polynomials, i.e., the arrangement
$\A(\F\cup\E)$ in $\R^d$, has the following properties:
\begin{enumerate}
\item[\rm(i)] {\rm (``Cylindrical'' cells) } 
For each cell $\sigma$ of $\A(\F\cup\E)$, there exists
a unique cell $\tau$ of the $(d-1)$-dimensional arrangement
$\A(\E)$ in $H$, such that one of the following possibilities occur:
\begin{enumerate}
\item[\rm (a)]  $\sigma=\{(x,\xi(x))\mid x\in\tau\}$, where
$\xi\colon\tau\to\R$ is a continuous semialgebraic function
(that is, $\sigma$ is the graph of $\xi$ over $\tau$).
\item[\rm (b)] $\sigma=\{(x,t)\mid x\in\tau, t\in (\xi_1(x),\xi_2(x))\}$,
where each $\xi_i$, $i=1,2$, is either a continuous semialgebraic 
real-valued function on $\tau$, or the constant function $\tau\to\{\infty\}$,
or the constant function $\tau\to\{-\infty\}$, and $\xi_1(x)<\xi_2(x)$
for all $x\in\tau$ (that is, $\sigma$ is a portion 
of the ``cylinder'' $\tau\times\R$ between two consecutive graphs).
\end{enumerate}
\item[\rm(ii)] {\rm (Refinement property) }
If $\F'\subseteq\F$, then $\E'=\E(\F')\subseteq\E$,
and thus each cell of $\A(\F\cup\E)$ is fully contained in
some cell of $\A(\F'\cup\E')$.
\end{enumerate}
\end{lemma}

Returning to the problem of decomposing the zero
set of the partitioning polynomial $f$ into monotone patches,
we construct the first-stage CAD for
$\F=\{f\}$, and the patches are the cells
of $\A(\F\cup\E)$ contained in $Z(f)$. 
If the $x_d$-direction is good for $f$, then every cell of 
$\A(\F\cup\E)$ lying in $Z(f)$ is of type~(a), and so if any cell of type~(b) 
lies in $Z(f)$, we choose another random direction and construct the 
first-stage CAD in that direction. Putting everything
together and using Theorem~\ref{t:make-arrg} to bound the
complexity of $\A(\F\cup\E)$, we obtain the following lemma.

\begin{lemma}
\label{l:patches}
Let $f$ be a $d$-variate polynomial of degree $D$,
 and let us assume that the $x_d$-direction is good for $f$.
Then $Z(f)$ can be decomposed, in $D^{O(d^4)}$ time, into
$D^{O(d)}$ monotone patches, and each patch can be 
represented semialgebraically by $D^{O(d^4)}$ polynomials of degree $D^{O(d^3)}$.
\end{lemma}

The first-stage CAD can also be used to compute the projection of 
the intersection of a range
in $\Gamma_{d,\Delta,s}$ with a monotone patch of $f$. 

\begin{lemma}
\label{l:project} 
Let $\Pi$ be the decomposition of the zero set of a $d$-variate
polynomial $f$ of degree $D$ into 
monotone patches, as described
in Lemma~\ref{l:patches}, and let $\gamma$ be a semialgebraic set in 
$\Gamma_{d,\Delta,s}$, with $\Delta \le D$. For every patch $\pi\in \Pi$,
the projection of $\gamma \cap \pi$ in the $x_d$-direction can be represented 
as a member of $\Gamma_{d-1,\Delta_1,s_1}$, 
\iffull 
i.e., by a Boolean combination
of at most $s_1$ polynomial inequalities in $(d-1)$ variables, each of
degree at most $\Delta_1$, 
\fi
where
$\Delta_1 = D^{O(d^3)}$ and $s_1=(Ds)^{O(d^4)}$.
The representation can be computed in $(Ds)^{O(d^4)}$ time.
\end{lemma}

\iffull
\begin{proof}
The task of computing $\gamma_\pi$,  the projection
of $\gamma\cap\pi$, is similar to the operation (A2) discussed
in Section~\ref{sec:cross}. In more abstract terms, 
it can also be viewed as a quantifier elimination task: 
we can represent $\gamma\cap \pi$ by a quantifier-free
formula $\Phi(x_1,\ldots,x_d)$ (a Boolean combination of polynomial
inequalities);
then $\gamma_\pi$ is represented
by $\exists x_d\Phi(x_1,\ldots,x_d)$, and by eliminating
$\exists x_d$ we obtain a quantifier-free formula describing $\gamma_\pi$.
More concretely, we use a procedure based on the first-stage CAD 
(Lemma~\ref{l:CAD}) and the arrangement construction (Theorem~\ref{t:make-arrg}).

By definition, $\gamma$ is a Boolean combination of inequalities
of the form $g_1\ge 0,\ldots,g_s\ge 0$,  where $g_1,\ldots,g_s$
are $d$-variate polynomials, each of degree at most $\Delta\le D$.
We set $\tilde\F:=\{f,g_1,\ldots,g_s\}$, we compute the
set $\tilde\E=\E(\tilde\F)$ of $(d-1)$-variate polynomials
as in Lemma~\ref{l:CAD}, and the first-stage CAD is
then computed as the $d$-dimensional arrangement $\A(\tilde\F\cup\tilde\E)$
according to Theorem~\ref{t:make-arrg}. 
Since by Lemma~\ref{l:CAD}(ii), $\A(\tilde\F\cup\tilde\E)$ refines 
$\A(\{f\}\cup\E(\{f\}))$ (the first-stage CAD from the
preprocessing phase), each patch $\pi\in\Patches$ is
decomposed into subpatches. Since the sign of each $g_i$
is constant on each cell of $\A(\tilde \F)$, and thus on each
cell of $\A(\tilde\F\cup\tilde\E)$, $\gamma\cap\pi$ is a disjoint
union of subpatches. The projections of these subpatches
into $H$ are cells of $\A(\tilde\E)$, and thus we obtain,
in time $(Ds)^{O(d^4)}$, a representation of $\gamma_\pi$ as a member of 
$\Gamma_{d-1,\Delta_1,s_1}$ by Theorem~\ref{t:make-arrg}, where 
$\Delta_1 = D^{O(d^3)}$ and $s_1 =(Ds)^{O(d^4)}$.
\end{proof}
\fi

\iffull
\section{Large Fan-Out Partition Tree: Proof of Theorem~\ref{t:large-r}}
\fi
\label{sec:range2}

We now describe our second data structure for ${\Gamma}_{d,\Delta,s}$-range 
searching. Compared to the first data structure from Section~\ref{sec:range1},
this one  works on arbitrary point sets, without the $D_0$-general position 
assumption, or, alternatively, without the fuzzy boundary constraint on the output, 
and has slightly better performance bounds.
The data structure is built recursively, and this time the
recursion involves both $n$ and~$d$. 

\subsection{The data structure}
Let $P$ be a set of $n$ points in $\R^d$, and let $\Delta$ and $s$ be
parameters (not assumed to be constant).
The data structure for
$\Gamma_{d,\Delta,s}$-range searching on $P$ is obtained by 
constructing a partition tree $\T$ on $P$ recursively, as above,
except that now the fan-out of each node is larger (and non-constant), and each node 
also stores an auxiliary data structure for handling the respective 
exceptional part. We need to set two parameters: $n_0=n_0(d,\Delta,s)$ 
and $r=r(d,\Delta,s,n)$. Neither of them is a constant in general;
in particular, $r$ is typically going to be a tiny power of~$n$.
The specific values of these parameters will be specified later,
when we analyze the query time.

\iffull
We also note that there is yet another parameter
in Theorem~\ref{t:large-r}, namely, the arbitrarily small
constant $\eps>0$ entering the preprocessing
time bound. However, $\eps$ enters the construction solely
by the requirement that $r$ should be chosen smaller than
$n^{\eps/c}$, for a sufficiently large constant~$c$.
It will become apparent later in the analysis that $r\le n^{\eps/c}$
can be assumed, provided that some other parameters are
taken sufficiently large; we will point this out at 
suitable moments.
\fi

When constructing the partition tree $\T$ on an $n$-point
set $P$, we distinguish two cases.  For $n\le n_0$, $\T$ consists
of a single leaf storing all points of $P$. 
For $n>n_0$,
we construct an $r$-partitioning polynomial $f$ of degree 
$D=O(r^{1/d})$, the partition 
$\thePartit$ of $\R^d$ induced by $f$, and the partition of $P$ into 
the exceptional part $P^*$ and regular parts $P_1,\ldots,P_t$, 
where $t=O(r)$. Set $n^*=|P^*|$ and $n_i = |P_i|$, 
for $i=1,\ldots,t$. The root of $\T$ stores $f$, $\thePartit$, and 
the total weight $w(P_i)$ of each regular part $P_i$ of $P$, 
as before.  Still in the same way as before, we recursively 
preprocess each regular part $P_i$ for $\Gamma_{d,\Delta,s}$-range 
searching (or stop if $|P_i|\le n_0$), 
and attach the resulting data structure to the root 
as a respective subtree. 

\heading{Handling the exceptional part.}
A new feature of the second data structure is that we also preprocess 
the exceptional set $P^*$ into an auxiliary data structure, 
which is stored at the root. Here we recurse on the dimension,
exploiting the fact that $P^*$ lies on the algebraic variety
$Z(f)$ of dimension at most $d-1$.

We choose a random direction $v$ and rotate the coordinate system so that
$v$ becomes the direction of the $x_d$-axis.
We construct the first-stage CAD adapted to $\{f\}$,
according to Lemma~\ref{l:CAD} and Theorem~\ref{t:make-arrg}.
We check whether all the patches are $x_d$-monotone, i.e.,
of type (a) in Lemma~\ref{l:CAD}(i); if it is not the case,
we discard the CAD and repeat the construction, with a different random 
direction.
This yields a decomposition of $Z(f)$ into a set $\Patches$ of 
$D^{O(d)}$ monotone patches, and the running time is $D^{O(d^4)}$ with 
high probability. 

Next, we  distribute the points of $P^*$ among the patches:
for each patch $\pi\in\Patches$,  let $P^*_\pi$ 
denote the projection of $P^*\cap\pi$ onto the coordinate hyperplane 
$H=\{x\in\R^d\,\mid\, x_d=0\}$.
We preprocess each set $P^*_\pi$ for 
$\Gamma_{d-1,\Delta_1,s_1}$-range searching. Here $s_1=(Ds)^{O(d^4)}$ is 
the number of polynomials defining a range and $\Delta_1=D^{O(d^3)}$ is
their maximum degree; the constants hidden 
in the $O(\cdot)$ notation are the same as in 
Lemma~\ref{l:project}.
\iffull
For simplicity, we treat all patches as being $(d-1)$-di\-mension\-al 
(although some may be of lower dimension); this 
does not influence the worst-case performance analysis.
\fi

The preprocessing of the sets $P^*_\pi$ is done recursively, 
using an $r_1$-partitioning polynomial in $\R^{d-1}$, for a 
suitable value of~$r_1$.
The exceptional set at each node of the resulting 
``$(d-1)$-di\-mension\-al'' tree is handled in a similar manner, 
constructing an auxiliary data structure in $d-2$ dimensions,
based on a first-stage CAD, and storing it 
at the corresponding node. The recursion on $d$ bottoms out at 
dimension $1$, where the structure is simply a standard binary 
search tree over the resulting set of points on the $x_1$-axis.
\iffull
  We remark that the treatment of the top level of 
  recursion on the dimension will be somewhat different from that of deeper 
  levels, in terms of both the choice of parameters and the 
  analysis; see below for details.
\fi

This completes the description of the data structure, except for the 
choice of $r$ and $n_0$,
which will be provided 
later as we analyze the 
performance of the algorithm.

\heading{Answering a query.} Let us assume that, for a given $P$,
the data structure for ${\Gamma}_{d,\Delta,s}$-range searching,
as described above, has been constructed, and 
consider a query range $\gamma \in {\Gamma}_{d,\Delta,s}$.
The query is answered in the same way as before, by visiting 
the nodes of the partition tree 
$\T$ in a top-down manner, except that, at each node
that we visit, we also query with $\gamma$ the auxiliary data 
structure constructed on the exceptional set $P^*$ for that node. 

Specifically, for each patch $\pi$ of the corresponding
collection $\Patches$, we compute 
$w_\pi$, the weight of $P^* \cap (\gamma\cap\pi)$.
If $\gamma\cap\pi = \emptyset$ then $w_\pi=0$, and
if $\gamma\cap\pi = \pi$ then $w_\pi$ is the total weight of 
$P^*\cap \pi$. 
Otherwise, i.e., if $\gamma$ crosses $\pi$, then 
$w_\pi$ is the same as the weight of $P^*_\pi \cap \gamma_\pi$,
where $\gamma_\pi$ is the $x_d$-projection of $\gamma\cap\pi$,
because $\pi$ is $x_d$-monotone.
By Lemma~\ref{l:project},
$\gamma_\pi \in \Gamma_{d-1,\Delta_1,s_1}$ and can be constructed in
$(Ds)^{O(d^4)}$ time. 
We can find the weight of $\gamma_\pi\cap P^*_\pi$ 
by querying the auxiliary data structure for $P^*_\pi$ with $\gamma_\pi$.
We then add $w_\pi$ to the global count maintained by the query procedure.
This completes the description of the query procedure.

\subsection{Performance analysis}
\iffull
The analysis of the storage requirement and preprocessing time
is straightforward, and will be provided later. We begin with
the more intricate analysis of the query time.
For now we assume that $n_0$ and $r$ have been fixed; the analysis will
later specify their values.
\fi

Let $\qtime_d(n,\Delta,s)$ denote the maximum overall query time
for ${\Gamma}_{d,\Delta,s}$-range searching on a set of $n$ 
points in $\R^d$. 
For $n \le n_0$ and $d \ge 1$, $\qtime_d(n,\Delta,s)=O(n)$. 
For $d=1$ and $n > n_0$, $\qtime_1(n,\Delta,s)=O(\Delta s\log n)$
because any range in $\Gamma_{1,\Delta,s}$
is the union of at most $\Delta s$
intervals. 
Finally, for $d > 1$ and $n > n_0$, an analysis similar to the one in Section~\ref{sec:range1}
gives the following recurrence for $\qtime_d(n,\Delta,s)$:
\iffull
\begin{equation}
\qtime_d(n,\Delta,s)\le  
C\Delta s r^{1-1/d} \qtime_d(n/r,\Delta, s) + 
  \sum_{\pi \in \Patches} \qtime_{d-1}(n_\pi,\Delta_1,s_1) + r^{c} ,
\label{eq:qtime}
\end{equation}
\fi
where $c=d^{O(1)}$, $C$ is a constant depending on~$d$,
$\sum_\pi n_\pi \le n$,  and both
$|\Patches|$ and  $\Delta_1s_1$  are bounded by $(Ds)^{a_d}$ 
with $D=O(r^{1/d})$ and $a_d=O(d^4)$.
(These are rather crude estimates, but we prefer simplicity.)
The leading term of the recurrence relies on 
the crossing-number bound given in Lemma~\ref{l:cr}.
In order to apply that lemma,
we need that $r\ge \Delta^d$, which will be ensured by the choice of
$r$ given below.
The second term corresponds to querying the auxiliary data structures
for the exceptional set $P^*$. The last term covers  
the time spent in computing the cells of the polynomial partition crossed
by the query range $\gamma$ and for computing the projections $\gamma_\pi$ 
for every $\pi \in \Pi$; here we assume that the choice of $r$ will be such that
$r \ge Ds$.

Ultimately, we want to derive that if $\Delta, s$ are constants, 
the recurrence (\ref{eq:qtime}) implies, 
with a suitable choice of $r$ and $n_0$ at each stage, 
\begin{equation} \label{qt:good}
\qtime_d(n,\Delta,s)\le n^{1-1/d}\log^{B(d,\Delta,s)}n ,
\end{equation}
where $B(d,\Delta,s)$ is a constant depending on $d, \Delta$, and $s$.

However, as was already mentioned, even if $\Delta, s$ are constants 
initially, later in the recursion 
they are chosen as tiny powers of $n$, 
and this makes it hard to obtain a direct 
inductive proof 
of~(\ref{qt:good}). Instead, we proceed in two stages.
\iffull
First, in Lemma~\ref{l:weakerQ} below we derive,  without assuming
$\Delta, s$ to be constants, a weaker bound for $\qtime_d(n,\Delta,s)$,
for which the induction is
easier. Then we obtain the stronger bound (\ref{qt:good}) for constant
values of $\Delta, s$
by using the weaker bound for the $(d-1)$-dimensional queries
on the exceptional parts, i.e., for the second term
in the recurrence~(\ref{eq:qtime}).

\heading{A weaker bound for lower-dimensional queries. }
\begin{lemma} \label{l:weakerQ}
For every $\auxeps>0$ there exists  $A_{d,\auxeps}$  such that,
with a suitable choice of $r$ and $n_0$,
\begin{equation}\label{qt:bad}
\qtime_d(n,\Delta,s) \le (\Delta s)^{A_{d,\auxeps}}
n^{1-1/d+\auxeps}
\end{equation}
for all $d,n,\Delta,s$ (with $\Delta s\ge 2$, say).
\end{lemma}

\heading{Remarks.} (i) This lemma may look similar to our first result
on $\Gamma_{d,\Delta,s}$-range searching, Theorem~\ref{t:const-r},
but there are two key differences---the lemma works for arbitrary
point sets, with no general position assumption, and
$\Delta$ and $s$ are not assumed to be constants.

\medskip
\noindent
(ii) Since query time $O(n)$ is trivial to achieve,
we may assume $\auxeps< 1/d$, for otherwise, the bound
(\ref{qt:bad}) in the lemma exceeds $n$. 
\medskip

\begin{proof}
The case $d=1$ is trivial because $Q_1(n,\Delta,s)\le C\Delta s\log_2 n$
clearly implies~(\ref{qt:bad}), assuming that
$A_{d,\auxeps} \ge 1+ \log_2 C$ and that $n$ is sufficiently large so that
$\log_2 n\le n^{\auxeps}$.
We assume that (\ref{qt:bad}) holds up to dimension $d-1$
(for all $\auxeps>0$, $\Delta$, $s$, and $n$), and we establish
it for dimension $d$ by induction on~$n$. We consider
$A_{d,\auxeps}$ yet unspecified but sufficiently large;
from the proof below one can obtain an explicit
lower bound that~$A_{d,\auxeps}$ should satisfy.
We set 
$$n_0=n_0(d,\Delta,s,\auxeps):= (\Delta s)^{d A_{d,\auxeps}}
\quad\mbox{and}\quad r=(2C\Delta s)^{1/\auxeps}.$$
This value of $n_0$ is roughly the threshold where the bound (\ref{qt:bad})
becomes smaller than $n$.
Since we assume $\auxeps<1/d$,
our choice of $r$ satisfies the assumptions $r\ge \Delta^d$ and
$r \ge Ds$, as needed in (\ref{eq:qtime}).

In the inductive step, for $n \le n_0$, 
$$Q_d(n,\Delta,s)  \le n \le n_0^{1/d}n^{1-1/d}
	= (\Delta s)^{A_{d,\auxeps}}n^{1-1/d} 
	\le (\Delta s)^{A_{d,\auxeps}}n^{1-1/d+\auxeps}.$$ 
So 
we assume that $n > n_0$ and that the bound (\ref{qt:bad}) holds for 
all $n'<n$. Using the induction hypothesis,
i.e., plugging (\ref{qt:bad}) into the recurrence (\ref{eq:qtime}), we obtain
\begin{align}
\qtime_d(n,\Delta,s) 
& \le C \Delta s r^{1-1/d}(\Delta s)^{A_{d,\auxeps}}
	(n/r)^{1-1/d+\auxeps} 
+ |\Pi| (\Delta_1s_1)^{A_{d-1,\auxeps}} n^{1-1/(d-1)+\auxeps} + r^{c}.
\label{eq:induct}
\end{align}
By the choice of $r$, the first term of the right-hand side
of (\ref{eq:induct}) can be bounded by
$$C\Delta s r^{-\auxeps} (\Delta s)^{A_{d,\auxeps}}n^{1-1/d+\auxeps}=
 \frac 12  (\Delta s)^{A_{d,\auxeps}}n^{1-1/d+\auxeps},$$ 
which is half of the bound we are aiming for.

Next, we bound the second term. We use the estimates 
$\Delta_1 s_1\le (Ds)^{a_d}$, $|\Patches|\le (Ds)^{a_d}$, and 
$Ds\le r$.
Then
\begin{eqnarray}
|\Pi|(\Delta_1s_1)^{A_{d-1,\auxeps}}
n^{1-1/(d-1)+\auxeps}&\le &
r^{a_d (A_{d-1,\auxeps}+1)} n^{1-1/(d-1)+\auxeps}\nonumber\\
&\le & \frac{r^{a_d (A_{d-1,\auxeps}+1)}}{ n^{1/d(d-1)}}\cdot n^{1-1/d+\auxeps}.
\label{eq:q-aux}
\end{eqnarray}
We choose 
\begin{equation}
A_{d,\auxeps} = \frac{d-1}{\auxeps}a' a_d(A_{d-1,\auxeps}+1),
\label{eq:Ad}
\end{equation}
where $a'=\log_2(4C)$;
i.e., we choose $A_{d,\auxeps}=d^{\Theta(d)}/\auxeps^d$.
Since $n\ge n_0= (\Delta s)^{d A_{d,\auxeps}}$ and $r=(2C\Delta s)^{1/\auxeps}$,
the fraction in (\ref{eq:q-aux}) can be bounded by
\begin{align*}
\frac{r^{a_d (A_{d-1,\auxeps}+1)}}{ n^{1/d(d-1)}} &\le
	\frac{(2C\Delta s)^{A_{d,\auxeps}/a'(d-1)}}{(\Delta
s)^{A_{d,\auxeps}/(d-1)}}
	\le \left ( \frac{2C}{(\Delta s)^{a'-1}}\right )^{A_{d,\auxeps}/a'(d-1)}
	\le 1
\end{align*}
because $\Delta s\ge 2$.

Finally, recall that $c=d^{O(1)}$, so our choice of $A_{d,\auxeps}$ 
(again, choosing $a'$ sufficiently large) ensures that $r^c < n^{1-1/d}$.
Hence, the right hand side in (\ref{eq:induct}) is bounded by 
$$
\tfrac12 (\Delta s)^{A_{d,\auxeps}}n^{1-1/d+\auxeps} +
2n^{1-1/d+\auxeps} \le
(\Delta s)^{A_{d,\auxeps}}n^{1-1/d+\auxeps} , 
$$
as desired.
This establishes the induction step and thereby completes the
proof of the lemma.
\end{proof}

\heading{The improved bound for the query time. }
Now we want to obtain the improved bound (\ref{qt:good}), i.e.,
$\qtime_d(n,\Delta,s)\le n^{1-1/d}\log^{B}n$, with $B=B(d,\Delta,s)$,
assuming that $\Delta, s$ are constants and $n> 2$.
To this end, in the top-level ($d$-dimensional) partition tree, 
we set $r:=n^\delta$, where $\delta>0$ is a suitable small 
constant to be specified later.
Then we  use the result of Lemma~\ref{l:weakerQ} with 
$\auxeps:=\frac 1{2d(d-1)}$
for processing the $(d-1)$-dimensional queries
on the sets $P^*_\pi$. Thus, in the
forthcoming proof, we do induction only on $n$, while
$d$ is fixed throughout. 

We choose $n_0=n_0(d,\Delta,s)$ sufficiently large (we will specify this more
precisely later on),
and we assume that $n>n_0$ and that the desired bound (\ref{qt:good}) holds
for all $n'<n$.
In the inductive step we estimate, using the 
recurrence (\ref{eq:qtime}), the induction hypothesis, and the bound
in (\ref{qt:bad}),
\[
\qtime_d(n,\Delta,s)  \le 
C\Delta s r^{1-1/d} (n/r)^{1-1/d} \log^{B}(n/r) 
+|\Pi|(\Delta_1 s_1)^{A_{d-1,\auxeps}} n^{1-1/(d-1)+\auxeps}
  + r^{c}.
\]

The first term simplifies to $(1-\delta)^B C\Delta sn^{1-1/d}\log^B n$.
Thus, if we choose $B$ depending on $\delta$ (which is a small
positive constant still to be determined) so that
$(1-\delta)^BC\Delta s\le \frac 12$, then the first term will be
at most half of the target value $n^{1-1/d}\log^B n$.
Thus, it suffices to set $\delta$ so that the remaining two terms
are negligible compared to this value.

For the $r^{c}$ term, any $\delta\le 1/2c$ will do.
The second term can be bounded, as in the proof
of Lemma~\ref{l:weakerQ}, by 
$$
r^{a_d(A_{d-1,\auxeps}+1)}n^{1-1/(d-1)+\auxeps} =
\frac{r^{\auxeps A_{d,\auxeps}/(a'(d-1))}} {n^\nu}\cdot n^{1-1/d}.
$$
Thus, with $\delta\le a'(d-1)/A_{d,\auxeps}$,
the term is at most $n^{1-1/d}$. 
Again, this establishes the induction step and
concludes the proof of the final bound for the query time. 
We remark that our choice of $\delta$ requires us to choose
$$
B \approx \frac{1}{\delta} \ln (2C\Delta s) \approx
\ln (2C\Delta s) \cdot d^{\Theta(d)} ,
$$
making its dependence on $d$ super-exponential.
\fi

\iffull
\heading{Analysis of storage and preprocessing.}
 Let $\storage_d(n,\Delta,s)$ denote the size of the data
structure on $n$ points in $\reals^d$ for ${\Gamma}_{d,\Delta,s}$-range 
searching, with the settings of $r$ and $n_0$ as described above.
For $n\le n_0=n_0(d,\Delta,s)$ we have $\storage_d(n,\Delta,s)=O(n)$. 
For larger values of $n$, the space occupied by the root 
of the partition tree, not counting the auxiliary
data structure for the exceptional part $P^*$, is bounded by
$r^{c}$, where $c=d^{O(1)}$.
Furthermore, since $\storage_d(n,\Delta,s)$ is at least linear in $n$, 
the total size of the  auxiliary data structure constructed on $P^*$ is 
$\sum_{\pi\in\Patches} \storage_{d-1}(n_\pi,\Delta_1,s_1) \le 
\storage_{d-1}(n^*,\Delta_1,s_1)$, where $n^*=|P^*|$. 
We thus obtain the following recurrence for $\storage_d(n,\Delta,s)$:
$$
\storage_d(n,\Delta,s) \le 
	\displaystyle \sum_{i=1}^t \storage_d(n_i,\Delta,s) +
        \storage_{d-1}(n^*,\Delta_1,s_1) + 
		O(r^{c}) 
$$
for $n>n_0=n_0(d,\Delta,s)$, and $\storage_d(n,\Delta,s)=O(n)$
for $n\le n_0$.
Using $n_i\le n/r$, $n^* + \sum_i n_i \le n$, and 
$r^{c} =o(n)$, for both types of choices of $r$, the recurrence 
easily leads to 
$$\storage_d(n,\Delta,s)= O(n) ,$$
where the constant of proportionality depends on $d$.

It remains to estimate the preprocessing time; here, finally,
the parameter $\eps>0$ in Theorem~\ref{t:large-r} comes into
play. Let $\delta^*$ be a constant such that $r\le n^{\delta^*}$
(at all stages of the algorithm). As was remarked in the preceding
analysis of the query time, we can make $\delta^*$ arbitrarily
small, by adjusting various constants (and, generally speaking
and as already remarked above,
the smaller $\delta^*$, the worse constant $B(d,\Delta,s)$
we obtain in the query time bound).

Let $\ptime_d(n,\Delta,s,\delta^*)$ denote the maximum preprocessing time
of our data structure for $\Gamma_{d,\Delta,s}$-range
searching on $n$ points, with $\delta^*>0$ a constant as above.
Using the operation (A1) of 
Section~\ref{sec:cross}, we spend $O(nr^{c})$ time 
to compute $\thePartit(f)$ and the partition of $P$ into the
exceptional part and the regular parts, and we
spend additional $O(nr^{c})$ time to compute $\Patches$ and $P^*_\pi$ 
for every $\pi \in \Patches$, where $c=d^{O(1)}$.
The total time spent in constructing 
the secondary data structures for all patches of $\Patches$ is bounded by 
$\ptime_{d-1}(n^*,\Delta_1,s_1,\delta^*)$.  Hence, we obtain the recurrence 
$$\ptime_d(n,\Delta,s,\delta^*)\le 
	\displaystyle \sum_{i=1}^t \ptime_d(n_i,\Delta,s,
\delta^*)+\ptime_{d-1}(n^*,\Delta_1,s_1,\delta^*) + 
	O(nr^{c}) 
$$
for $n>n_0$, and $\ptime_d(n,\Delta,s,\delta^*)=O(n)$ for $n\le n_0$.
Using the properties  $n_i\le n/r$ and $n^* + \sum_i n_i \le n$, 
a straightforward calculation shows that
$$\ptime_d(n,\Delta,s,\delta^*) = O(n^{1+c\delta^*}) ,$$
where the constant of proportionality depends on $d$.
Hence, by choosing $\delta^*=\eps/c$, the preprocessing time 
is $O(n^{1+\eps})$.
This concludes the proof of Theorem~\ref{t:large-r}.
\fi

\section{Open Problems}
\label{sec:concl}

We conclude this paper by mentioning a few open problems.

\medskip

\noindent (i) 
A very interesting and challenging problem is, in our opinion,
the fast-query case of range searching with constant-complexity
semialgebraic sets, where the goal is to answer a query in $O(\log n)$ time using
roughly $n^d$ space. 
There are actually two, apparently distinct, issues.
The standard approach to fast-query searching is to parameterize
the ranges in $\Gamma$ by points in a space of a suitable dimension, say $t$;
then the $n$ points of $P$ correspond to $n$  algebraic surfaces in this
$t$-dimensional ``parameter space'', and a query is answered by locating the 
point corresponding to the query range in the arrangement of these
surfaces.

First, the arrangement has $O(n^t)$ combinatorial complexity,
and one would expect to be able to locate points in
it in polylogarithmic time with storage about $n^t$.
However, such a method is known only up to dimension $t=4$,
and in higher dimension, one again gets stuck at the 
arrangement decomposition problem, which was the bottleneck
in the previously known solution of \cite{AM} for the
low-storage variant, as was mentioned in the introduction. It would be
nice to use polynomial partitions to obtain a better
point location data structure for such arrangements,
but unfortunately, so far all of our attempts in this direction
have failed.

The second issue is, whether the point location approach
just sketched is actually optimal. This question is exhibited
nicely already in the simple instance of range searching with disks in
the plane.  The best known solution that
guarantees logarithmic query time uses point
location in $\R^3$ and  requires  storage roughly $n^3$,
but it is conceivable that roughly quadratic storage might
suffice.

\bigskip

\noindent (ii) 
Our range-searching data structure for arbitrary point
sets---the one with large fan-out---is so complex and has
a rather high exponent in the polylogarithmic factor, because we
have difficulty with handling highly degenerate point sets,
where many points lie on low-degree algebraic surfaces. 
This issue
appears even more strongly in combinatorial applications,
and in that setting it has been dealt with only in rather
specific cases (e.g., in dimension 3); see \cite{KMSS,SoTa,Za} 
for  initial studies.  
It would be nice to find a
construction of suitable ``multilevel polynomial partitions''
that would cater to such highly degenerate input sets,
as touched upon in \cite{KMSS,Za}.

\bigskip

\noindent (iii) 
Another open problem, related to the construction
of polynomial partitions, is the fast evaluation of a 
multivariate polynomial at many points, as briefly discussed at 
the end of Section~\ref{sec:algo}.

\paragraph{Acknowledgments.} We thank the anonymous referees for their useful comments on the paper.


\end{document}